\documentclass[11pt]{article}

\usepackage[margin=1in]{geometry}
\usepackage{amsmath,amssymb,amsthm,mathtools}
\usepackage{bm}
\usepackage{booktabs}
\usepackage{enumitem}
\usepackage{graphicx}
\usepackage{array}
\usepackage{placeins}
\usepackage[authoryear,round]{natbib}
\usepackage{setspace}
\usepackage{microtype}
\usepackage{xurl}
\usepackage[colorlinks=true,linkcolor=blue,citecolor=blue,urlcolor=blue]{hyperref}

\newtheorem{assumption}{Assumption}
\newtheorem{theorem}{Theorem}
\newtheorem{proposition}{Proposition}
\newtheorem{corollary}{Corollary}

\newtheorem{lemma}{Lemma}

\newcommand{\E}{\mathbb{E}}
\newcommand{\Pp}{\mathbb{P}}
\newcommand{\R}{\mathbb{R}}
\newcommand{\Var}{\operatorname{Var}}
\newcommand{\Cov}{\operatorname{Cov}}

\newcommand{\argmin}{\operatorname*{arg\,min}}
\newcommand{\norm}[1]{\left\lVert #1 \right\rVert}
\newcommand{\ip}[2]{\left\langle #1,#2\right\rangle}
\newcommand{\Pn}{\mathbb{P}_n}

\newcommand{\ind}{\mathbf 1}

\title{\textbf{Information Borrowing for Cox Regression with an Auxiliary Outcome}}

\author{
Xiang Wang\thanks{School of Mathematics and Statistics, Shanxi University, Taiyuan, Shanxi, China.}
\and
Kexuan Li\thanks{Bristol Myers Squibb, Cambridge, Massachusetts, USA.}
\and
Lingli Yang\thanks{Takeda Pharmaceuticals, Cambridge, Massachusetts, USA.}
\and
Weidong Ma\thanks{Department of Biostatistics, University of Pennsylvania, Philadelphia, Pennsylvania, USA.}
\thanks{Corresponding author: \texttt{weidong.ma@pennmedicine.upenn.edu}}
}

\date{}

\begin{document}
\maketitle
\begin{abstract}
Time-to-event outcomes are often collected together with auxiliary outcomes that may provide additional information about baseline risk. We develop a semiparametric approach for incorporating such information into Cox regression without specifying a joint likelihood. The survival outcome follows a Cox proportional hazards model and a continuous auxiliary outcome follows a partially linear model, with separate nonparametric baseline covariate effects linked through a quadratic penalty. When these effects coincide and the outcome scores satisfy the information identities and a first-order orthogonality condition, we derive the sandwich covariance of the penalized estimator and show that, for any fixed penalty level, the asymptotic variance of the Cox regression estimator is no greater than that under separate estimation. We further characterize departures from the shared-effect setting. Local differences of order \(n^{-1/2}\) induce an explicit mean shift in the limiting distribution, yielding a direct bias--variance trade-off, whereas fixed differences generally alter the population target under nonvanishing penalization. These results motivate an adaptive penalty that borrows information when the fitted covariate effects are close and approaches separate estimation when a persistent difference is detected. Simulations and a real world data analysis demonstrate the validity
and effectiveness of the proposed method.
\end{abstract}

\noindent\textbf{Keywords:}
auxiliary outcome; information borrowing; penalized estimation; semiparametric inference;

\clearpage
\section{Introduction}
\label{sec:intro}

Survival outcomes are often collected together with other outcomes measured on the same individuals. In a clinical trial, for example, participants may have laboratory measurements, clinical scores, imaging summaries, or short-term responses in addition to the primary time-to-event outcome. Although these additional outcomes are not themselves the primary endpoint, they may contain useful information about baseline prognosis. This raises a natural question: can an auxiliary outcome improve inference for a regression parameter in a survival model by improving estimation of the baseline covariate effects used for adjustment?

Several related problems have been studied in the literature. Joint models for longitudinal and survival data link repeated measurements and event times through shared random effects or other latent structures
\citep{tsiatisEtAl1995,wulfsohnTsiatis1997,hendersonEtAl2000,songEtAl2002,tsiatisDavidian2004}.
Surrogate endpoint methods ask whether an intermediate outcome can replace or predict a clinical endpoint \citep{prentice1989}.
Auxiliary data have also been used to improve semiparametric estimation in missing-data and validation-sample problems
\citep{chenEtAl2003,chenEtAl2008,chenHongTarozzi2008}.
For survival outcomes, external summaries, subgroup information, and related source populations have been incorporated into Cox regression to improve estimation efficiency
\citep{huangEtAl2016,huangQin2020,shengEtAl2022,liEtAl2023}.
These approaches address different forms of information sharing. Our setting is one in which the survival and auxiliary outcomes are observed on the same individuals, and the goal is to use similarities in their relationships with baseline covariates without imposing a full joint model.

Let $Y$ denote a continuous auxiliary outcome, $T$ the event time, $C$ the censoring time, $Z\in\mathbb R^p$ a low-dimensional covariate vector containing the variables of primary scientific interest, and $X$ a vector of baseline covariates used for adjustment. We consider
\begin{align}
Y
&=
Z^\top\alpha_0 + f_{Y0}(X) + \varepsilon,
\label{eq:introaux}\\
\lambda(t\mid Z,X)
&=
\lambda_0(t)\exp\{Z^\top\beta_0+f_{T0}(X)\},
\label{eq:introcox}
\end{align}
where $\alpha_0\in\mathbb R^p$ is the auxiliary regression coefficient, $\varepsilon$ satisfies $\E(\varepsilon\mid Z,X)=0$, $\lambda_0(t)$ is an unspecified baseline hazard function, and $f_{Y0}$ and $f_{T0}$ are unknown functions describing the effects of the baseline covariates on the two outcomes. The parameter of interest is $\beta_0\in\mathbb R^p$. In a randomized trial, for example, $Z$ may include a treatment indicator, in which case a component of $\beta_0$ is the treatment log-hazard ratio.

We do not require the two outcomes to have the same treatment effect, nor do we require their baseline covariate effects to be identical. Instead, the two regression functions are estimated jointly with a quadratic penalty that encourages similarity between them. Specifically, if $f_Y$ and $f_T$ denote candidate regression functions for the auxiliary and survival outcomes, respectively, we use a penalty proportional to
\[
\|f_Y-f_T\|_n^2,
\qquad
\|g\|_n^2=\frac{1}{n}\sum_{i=1}^n g^2(X_i).
\]
The amount of information sharing is controlled by a nonnegative tuning parameter. With no penalization, the two outcomes are analyzed separately. As the penalty increases, information from the auxiliary outcome has a greater influence on estimation of the survival covariate function, while the two functions are still allowed to differ unless equality is imposed in the limit. Because $f_Y$ and $f_T$ are measured on different scales, their difference is meaningful only after the relative scale of the two outcomes has been fixed. We therefore standardize the auxiliary outcome before fitting the model and use equality of the two baseline covariate functions as a reference case for studying information sharing. A proportional formulation, in which the two functions may differ by an unknown multiplicative factor, is discussed in Appendix~\ref{app:rho}.

The remainder of the paper is organized as follows. Section~\ref{sec:model} introduces the model and estimation procedure. Section~\ref{sec:theory} presents the asymptotic results. Sections~\ref{sec:simulation} and \ref{sec:application} contain the simulation study and the ACTG 175 analysis, respectively. Proofs and additional technical details are given in the appendices.

\section{Model and estimation}
\label{sec:model}

\subsection{Outcome models}

Let
\[
O_i=(Y_i,\widetilde T_i,\Delta_i,Z_i,X_i),
\qquad i=1,\ldots,n,
\]
be independent and identically distributed observations, where $Y_i$ is the auxiliary outcome, $T_i$ is the event time, $C_i$ is the censoring time,
\[
\widetilde T_i=\min(T_i,C_i),
\qquad
\Delta_i=\ind(T_i\le C_i),
\]
and $Z_i\in\mathbb R^p$ and $X_i$ denote, respectively, the covariates of primary interest and the baseline covariates used for adjustment.

For subject $i$, define the counting process and at-risk process by
\[
N_i(t)=\ind(\widetilde T_i\le t,\Delta_i=1),
\qquad
R_i(t)=\ind(\widetilde T_i\ge t).
\]
When no subject index is needed, $N(t)$ and $R(t)$ denote the corresponding processes for a generic observation. We assume conditional independent censoring,
\[
T\perp C\mid(Z,X).
\]

The auxiliary outcome is modeled by
\begin{equation}
Y
=
Z^\top\alpha_0+f_{Y0}(X)+\varepsilon,
\qquad
\E(\varepsilon\mid Z,X)=0,
\label{eq:auxmodel}
\end{equation}
where $\alpha_0\in\mathbb R^p$ and $f_{Y0}$ is an unknown function of the baseline covariates. For the event time, we assume the proportional hazards model of \citet{cox1972},
\begin{equation}
\lambda(t\mid Z,X)
=
\lambda_0(t)
\exp\{Z^\top\beta_0+f_{T0}(X)\},
\label{eq:coxmodel}
\end{equation}
where $\lambda_0(t)$ is an unspecified baseline hazard, $\beta_0\in\mathbb R^p$ is the regression parameter of interest, and $f_{T0}$ is an unknown baseline covariate effect.

We impose the centering conditions
\begin{equation}
\E\{f_{Y0}(X)\}
=
\E\{f_{T0}(X)\}
=
0.
\label{eq:center}
\end{equation}

\subsection{Scale of the auxiliary outcome}
\label{subsec:scale}

Because $f_Y$ and $f_T$ are measured on different scales, their difference
is interpreted on a scale chosen before fitting the model. The function
$f_T$ is on the log-hazard scale, whereas $f_Y$ is on the scale of the
auxiliary outcome. The theoretical results treat this outcome scale as
fixed. Scaling the auxiliary squared-error criterion by the residual
variance is a separate operation, described in
Section~\ref{subsec:estimation}.

Under this normalization, equality of the two covariate functions provides a
reference case for studying information sharing. More generally, one may allow
the two functions to be proportional,
\[
f_{Y0}(x)=\rho_0 f_{T0}(x),
\]
where $\rho_0$ is a proportionality constant, and use the penalty
\[
\frac{\lambda}{2}\norm{f_Y-\rho f_T}_n^2,
\]
where $\lambda\ge0$ controls the penalty and $\rho$ is the candidate
proportionality parameter. This extension is discussed in
Appendix~\ref{app:rho}. The main results compare the two functions directly
on the chosen scale, with equality serving as the reference case.

\subsection{Penalized estimation}
\label{subsec:estimation}

Let $\sigma_Y^2>0$ denote the residual variance used to scale the
auxiliary criterion. Under the homoskedastic auxiliary model,
$\E(\varepsilon^2\mid Z,X)=\sigma_Y^2$. We first treat $\sigma_Y^2$
as fixed and define
\begin{equation}
L_{Y,n}(\alpha,f_Y)
=
\frac{1}{2n\sigma_Y^2}
\sum_{i=1}^n
\{Y_i-Z_i^\top\alpha-f_Y(X_i)\}^2.
\label{eq:auxloss}
\end{equation}
A consistent estimate of $\sigma_Y^2$ may be used in implementation.
When this variance or the outcome normalization is estimated from the
same data, its contribution to the first-order expansion must be included
unless it is shown to be asymptotically negligible.

For the survival outcome, let $0<\tau<\infty$ denote a fixed
analysis horizon. Up to a term independent of the parameters, the
negative Cox partial log-likelihood is
\begin{equation}
L_{T,n}(\beta,f_T)
=
-\frac1n\sum_{i=1}^n\int_0^\tau
\left[
Z_i^\top\beta+f_T(X_i)
-\log S_n^{(0)}(t;\beta,f_T)
\right]dN_i(t),
\label{eq:coxloss}
\end{equation}
where
\[
S_n^{(0)}(t;\beta,f)
=
\frac1n\sum_{j=1}^n
R_j(t)\exp\{Z_j^\top\beta+f(X_j)\}.
\]
For a function $g$, write
\[
\norm{g}_n^2
=
\frac1n\sum_{i=1}^n g^2(X_i).
\]

Let $\mathcal F_n$ be a centered sieve class satisfying
\[
\frac1n\sum_{i=1}^n f(X_i)=0,
\qquad f\in\mathcal F_n.
\]
For $\lambda\ge0$, define
\begin{equation}
Q_n(\alpha,\beta,f_Y,f_T;\lambda)
=
L_{Y,n}(\alpha,f_Y)
+
L_{T,n}(\beta,f_T)
+
\frac{\lambda}{2}\norm{f_Y-f_T}_n^2.
\label{eq:criterion}
\end{equation}
Given a penalty level $\lambda_n\ge0$, the proposed estimator is
\begin{equation}
(\widehat\alpha_{\lambda_n},
 \widehat\beta_{\lambda_n},
 \widehat f_{Y,\lambda_n},
 \widehat f_{T,\lambda_n})
\in
\argmin_{\substack{
\alpha,\beta\in\mathbb R^p\\
f_Y,f_T\in\mathcal F_n}}
Q_n(\alpha,\beta,f_Y,f_T;\lambda_n).
\label{eq:estimator}
\end{equation}
When $\lambda_n=0$, the two outcomes are fitted separately. Positive
values of $\lambda_n$ encourage similarity between the two estimated
covariate functions while allowing them to remain distinct.

For later use, let $L_Y$ and $L_T$ denote the population counterparts
of the two endpoint-specific criteria. For a fixed $\lambda\ge0$, the
corresponding population criterion is
\begin{equation}
Q(\alpha,\beta,f_Y,f_T;\lambda)
=
L_Y(\alpha,f_Y)
+
L_T(\beta,f_T)
+
\frac{\lambda}{2}\norm{f_Y-f_T}_{P_X}^2,
\label{eq:popQ}
\end{equation}
where $P_X$ denotes the distribution of $X$ and
\[
\norm{g}_{P_X}^2=\E\{g^2(X)\},
\qquad
\ip{g}{h}_{P_X}=\E\{g(X)h(X)\}.
\]

The criterion in \eqref{eq:criterion} is not a fully specified joint
likelihood for $(Y,T)$. Inference is based on the sandwich covariance
associated with the joint estimating equations. This
formulation allows the auxiliary and survival outcomes from the same
individual to be dependent. The variance comparison in
Section~\ref{sec:exact} imposes an additional first-order orthogonality
condition on their score contributions; this condition does not require
conditional independence of the two outcomes. A sandwich covariance
describes sampling variation around the relevant population target; it
does not correct a change in that target caused by penalization.

\subsection{Population target under penalization}
\label{sec:population}

We examine whether penalization preserves the population target when the
two covariate functions are equal and when they differ.

\subsubsection{Equal baseline covariate effects}

Suppose first that
\begin{equation}
f_{Y0}=f_{T0}=f_0.
\label{eq:exactsharing}
\end{equation}
In this case the penalty vanishes at the true pair of covariate
functions, so penalization does not alter the population target.

\begin{proposition}
\label{prop:exacttarget}
Suppose $(\alpha_0,f_0)$ minimizes $L_Y$ and $(\beta_0,f_0)$ minimizes
$L_T$ over their identifiable parameter spaces. If
\eqref{eq:exactsharing} holds, then, for every fixed $\lambda\ge0$,
\[
(\alpha_0,\beta_0,f_0,f_0)
\]
minimizes $Q(\alpha,\beta,f_Y,f_T;\lambda)$. If the endpoint-specific
minimizers are unique, then the minimizer of $Q$ is also unique.
\end{proposition}

Thus, when the two baseline covariate functions coincide, information
sharing through the penalty does not change the population target. Its effect is
instead on the sampling variability of the estimator, which is studied
in Section~\ref{sec:exact}.

\subsubsection{Fixed differences}

Now suppose that
\[
d_0=f_{Y0}-f_{T0},
\qquad \norm{d_0}_{P_X}>0.
\]
Let
\[
\mathcal P(f_Y,f_T)
=
\frac{\lambda}{2}\norm{f_Y-f_T}_{P_X}^2
\]
denote the population penalty. Its directional derivative at
$(f_{Y0},f_{T0})$ in a direction $(h_Y,h_T)$ is
\begin{equation}
D\mathcal P(f_{Y0},f_{T0})[h_Y,h_T]
=
\lambda\ip{d_0}{h_Y-h_T}_{P_X}.
\label{eq:penalty-direction}
\end{equation}
Unlike the equal-effect case, this derivative need not vanish at the
unpenalized truth.

\begin{proposition}
\label{prop:fixedpop}
Suppose $L_Y$ and $L_T$ are differentiable along admissible directions
at their respective interior minimizers $(\alpha_0,f_{Y0})$ and
$(\beta_0,f_{T0})$, where interiority is relative to the identifying
constraints. Let $\lambda>0$.
If there exists an admissible direction $(h_Y,h_T)$ such that
\[
\ip{d_0}{h_Y-h_T}_{P_X}\neq0,
\]
then
\[
(\alpha_0,\beta_0,f_{Y0},f_{T0})
\]
is not a stationary point of the penalized population criterion.
\end{proposition}

For the linear function spaces considered below, this condition is
satisfied whenever $d_0$ is an admissible direction. For example, taking
$h_Y=d_0$ and $h_T=0$ gives
\[
\ip{d_0}{h_Y-h_T}_{P_X}
=
\norm{d_0}_{P_X}^2
>0.
\]
Thus, when the two baseline covariate functions differ by a fixed
amount, a fixed positive penalty generally changes the population
target. This motivates allowing the amount of borrowing to decrease
when a persistent difference between the two functions is present.

\section{Asymptotic properties}
\label{sec:theory}

\subsection{Fixed-sieve expansion}
\label{sec:firstorder}

We begin with a fixed-dimensional sieve, for which the contribution of
auxiliary information can be described through the information matrices
of the two outcomes. The extension to a growing sieve is given in
Section~\ref{sec:growing}. Proofs are given in
Appendix~\ref{app:proofs}, with additional matrix calculations in
Appendix~\ref{app:algebra}.

\subsubsection{Sieve representation and score functions}

Let $b_K(x)\in\R^K$ be a fixed basis for the identifiable nuisance space
and write
\begin{equation}
f_Y(x)=b_K(x)^\top\theta_Y,
\qquad
f_T(x)=b_K(x)^\top\theta_T,
\label{eq:sieve}
\end{equation}
where $\theta_Y,\theta_T\in\R^K$. For the fixed-sieve results, we assume
that the true functions have this representation, with coefficients
$\theta_{Y0}$ and $\theta_{T0}$. Approximation error is considered in
Section~\ref{sec:growing}.

Write $\Pn q=n^{-1}\sum_{i=1}^n q(O_i)$ for an empirical average and define
\[
G_n=\Pn\{b_K(X)b_K(X)^\top\},
\qquad
G=\E\{b_K(X)b_K(X)^\top\}.
\]
The penalty in \eqref{eq:criterion} becomes
\begin{equation}
\frac{\lambda}{2}
(\theta_Y-\theta_T)^\top G_n(\theta_Y-\theta_T).
\label{eq:penaltytheta}
\end{equation}

For the survival model, let
$W=(Z^\top,b_K(X)^\top)^\top\in\R^{p+K}$ and define the population
risk-set moments
\[
s^{(r)}(t)
=
\E\left\{
R(t)\exp\bigl(Z^\top\beta_0+b_K(X)^\top\theta_{T0}\bigr)
W^{\otimes r}
\right\},
\qquad r=0,1,2,
\]
where $W^{\otimes0}=1$, $W^{\otimes1}=W$, and
$W^{\otimes2}=WW^\top$. Put
$\bar W(t)=s^{(1)}(t)/s^{(0)}(t)$ and
$\Lambda_0(t)=\int_0^t\lambda_0(u)\,du$. Under the Cox model and
conditional independent censoring, the event martingale is
\[
M(t)
=
N(t)-\int_0^t
R(u)\exp\bigl(Z^\top\beta_0+b_K(X)^\top\theta_{T0}\bigr)
\,d\Lambda_0(u).
\]
The limiting contribution of one observation to the Cox score is
\begin{equation}
s_S(O)
=
\int_0^\tau\{W-\bar W(t)\}\,dM(t)
=
\begin{pmatrix}
s_\beta(O)\\
s_T(O)
\end{pmatrix}.
\label{eq:coxscore}
\end{equation}
Here $s_\beta(O)\in\R^p$ and $s_T(O)\in\R^K$ correspond to $\beta$ and
$\theta_T$, respectively. Partition the population Cox information as
\begin{equation}
\mathcal I_T
=
\begin{pmatrix}
A&C\\
C^\top&D
\end{pmatrix},
\label{eq:IT}
\end{equation}
where $A\in\R^{p\times p}$, $C\in\R^{p\times K}$, and
$D\in\R^{K\times K}$.

For the auxiliary model, the population Hessian of the scaled
squared-error loss is
\[
H_Y
=
\frac{1}{\sigma_Y^2}\E(WW^\top)
=
\begin{pmatrix}
H_{ZZ}&H_{Zb}\\
H_{bZ}&H_{bb}
\end{pmatrix}.
\]
When $H_{ZZ}$ is nonsingular, profiling out $\alpha$ gives the Hessian
for $\theta_Y$,
\begin{equation}
E=H_{bb}-H_{bZ}H_{ZZ}^{-1}H_{Zb}.
\label{eq:Edef}
\end{equation}
Define the residualized basis
\begin{equation}
\widetilde b_K(X,Z)
=
\frac{1}{\sigma_Y}
\left\{b_K(X)-H_{bZ}H_{ZZ}^{-1}Z\right\}.
\label{eq:btilde}
\end{equation}
The corresponding profiled score is
\begin{equation}
s_Y(O)
=
\frac{\varepsilon}{\sigma_Y}\widetilde b_K(X,Z).
\label{eq:auxprofilescore}
\end{equation}
Under conditional homoskedasticity,
$\E(\varepsilon^2\mid Z,X)=\sigma_Y^2$, this score satisfies
\begin{equation}
\Var(s_Y)=E.
\label{eq:auxidentity}
\end{equation}
Without homoskedasticity, a constant scaling of the squared-error loss
need not give \eqref{eq:auxidentity}. A sandwich covariance can still be
used, but the variance ordering established below requires additional
conditions on the score covariance.

\subsubsection{Regularity conditions and limiting distribution}

After profiling out $\alpha$, write
\[
\vartheta=(\beta^\top,\theta_T^\top,\theta_Y^\top)^\top,
\qquad
\vartheta_0=(\beta_0^\top,\theta_{T0}^\top,\theta_{Y0}^\top)^\top,
\]
and let $\widehat\vartheta_\lambda$ denote the corresponding estimator.
We use the following conditions.

\begin{assumption}[Sampling and risk sets]
\label{ass:sampling}
The observations are independent and identically distributed, and
$T\perp C\mid(Z,X)$. The observation horizon satisfies $\tau<\infty$,
$\Pp(\widetilde T\ge\tau)>0$, and $\Lambda_0(\tau)<\infty$. The function
$s^{(0)}(t)$ is bounded away from zero on $[0,\tau]$.
\end{assumption}

\begin{assumption}[Nonsingularity]
\label{ass:pd}
For fixed $K$, the matrices $H_Y$, $\mathcal I_T$, and $G$ are finite,
$H_{ZZ}$ is nonsingular, and $G$, $E$, and $\mathcal I_T$ are positive
definite.
\end{assumption}

\begin{assumption}[Information identities]
\label{ass:identity}
The outcome scores satisfy
\[
\Var(s_S)=\mathcal I_T,
\qquad
\Var(s_Y)=E.
\]
\end{assumption}

\begin{assumption}[Score orthogonality]
\label{ass:crossorth}
The profiled auxiliary score is uncorrelated with the Cox score:
\[
\Cov(s_Y,s_\beta)=0,
\qquad
\Cov(s_Y,s_T)=0.
\]
\end{assumption}

\begin{assumption}[Local expansion]
\label{ass:quad}
For each fixed $K$ and $\lambda$, the profiled estimator is consistent
for the relevant population minimizer. The profiled estimating equations
admit a Taylor expansion about that minimizer, with an
$o_p(n^{-1/2})$ remainder when evaluated at the estimator. The empirical
derivative matrix converges in probability to its nonsingular population
counterpart, and the empirical unpenalized score has an asymptotically
linear representation satisfying a joint central limit theorem.
\end{assumption}

A sufficient condition for Assumption~\ref{ass:crossorth} is
\begin{equation}
\E\{\varepsilon\,dM(t)\mid Z,X\}=0,
\qquad t\in[0,\tau],
\label{eq:momentorth}
\end{equation}
with the integrability needed to interchange expectation and integration.
This is a moment restriction on the score contributions, rather than an
assumption of conditional independence between the outcomes.

\begin{lemma}
\label{lem:orth}
Suppose \eqref{eq:momentorth} holds and the score contributions are
square-integrable. Then the profiled auxiliary score in
\eqref{eq:auxprofilescore} is uncorrelated with both components of the
Cox score in \eqref{eq:coxscore}.
\end{lemma}

Under equal covariate effects,
$\theta_{Y0}=\theta_{T0}=\theta_0$, the penalty gradient vanishes at the
truth. The penalty contributes to the derivative matrix but not to the
random score at that point. The derivative and score covariance matrices
are therefore
\begin{equation}
\mathcal A_\lambda
=
\begin{pmatrix}
A&C&0\\
C^\top&D+\lambda G&-\lambda G\\
0&-\lambda G&E+\lambda G
\end{pmatrix}
\label{eq:Alambda}
\end{equation}
and
\begin{equation}
\mathcal B
=
\Var\begin{pmatrix}s_\beta\\s_T\\s_Y\end{pmatrix}
=
\begin{pmatrix}
A&C&0\\
C^\top&D&0\\
0&0&E
\end{pmatrix}.
\label{eq:B}
\end{equation}

\begin{theorem}
\label{thm:rootfixed}
Suppose $f_{Y0}=f_{T0}$ and
Assumptions~\ref{ass:sampling}--\ref{ass:quad} hold. For each fixed
$\lambda\ge0$,
\begin{equation}
\sqrt n(\widehat\vartheta_\lambda-\vartheta_0)
=
\mathcal A_\lambda^{-1}
\frac{1}{\sqrt n}\sum_{i=1}^n
\begin{pmatrix}
s_\beta(O_i)\\s_T(O_i)\\s_Y(O_i)
\end{pmatrix}
+o_p(1).
\label{eq:rootfull}
\end{equation}
Consequently,
\[
\sqrt n(\widehat\vartheta_\lambda-\vartheta_0)
\rightsquigarrow
N\{0,\mathcal A_\lambda^{-1}\mathcal B\mathcal A_\lambda^{-T}\}.
\]
The asymptotic covariance of
$\sqrt n(\widehat\beta_\lambda-\beta_0)$ is the $(\beta,\beta)$ block of
this sandwich matrix.
\end{theorem}

\subsection{Efficiency under equal baseline covariate effects}
\label{sec:exact}

Let $V_\lambda$ denote the asymptotic covariance of
$\sqrt n(\widehat\beta_\lambda-\beta_0)$ in
Theorem~\ref{thm:rootfixed}, and let $V_0$ denote its counterpart under
separate estimation. To compare these matrices, we eliminate the
auxiliary nuisance coefficient from the linearized equations. Define
\begin{align}
W_\lambda
&=\lambda G(E+\lambda G)^{-1},
\label{eq:Wlambda}\\
K_\lambda
&=\lambda G-\lambda^2G(E+\lambda G)^{-1}G,
\label{eq:Klambda}\\
H_\lambda
&=D+K_\lambda,
\label{eq:Hlambda}\\
P_\lambda
&=A-CH_\lambda^{-1}C^\top,
\label{eq:Plambda}\\
J_\lambda
&=D+W_\lambda E W_\lambda^\top,
\label{eq:Jlambda}\\
\Delta_\lambda
&=H_\lambda-J_\lambda.
\label{eq:Deltalambda}
\end{align}
Here and below, $\preceq$ denotes the positive-semidefinite order.

\begin{lemma}
\label{lem:eliminate}
Under the conditions of Theorem~\ref{thm:rootfixed}, eliminating the
$\theta_Y$ coordinate gives a system for $(\beta,\theta_T)$ with
derivative matrix
\begin{equation}
\widetilde{\mathcal A}_\lambda
=
\begin{pmatrix}
A&C\\
C^\top&H_\lambda
\end{pmatrix}
\label{eq:Atilde}
\end{equation}
and score contribution
\begin{equation}
\widetilde s_\lambda
=
\begin{pmatrix}
s_\beta\\
s_T+W_\lambda s_Y
\end{pmatrix}.
\label{eq:stilde}
\end{equation}
Its covariance is
\begin{equation}
\widetilde{\mathcal B}_\lambda
=
\Var(\widetilde s_\lambda)
=
\begin{pmatrix}
A&C\\
C^\top&J_\lambda
\end{pmatrix}.
\label{eq:Btilde}
\end{equation}
\end{lemma}

\begin{lemma}
\label{lem:psd}
If $E$ and $G$ are positive definite, then, for every $\lambda\ge0$,
\[
K_\lambda\succeq0,
\qquad
\Delta_\lambda\succeq0.
\]
These conclusions do not require $E$ and $G$ to commute.
\end{lemma}

The next theorem gives the covariance comparison for the survival
regression parameter.

\begin{theorem}
\label{thm:efficiency}
Under Assumptions~\ref{ass:sampling}--\ref{ass:quad} and
$f_{Y0}=f_{T0}$,
\begin{equation}
V_\lambda
=
P_\lambda^{-1}
-
P_\lambda^{-1}CH_\lambda^{-1}
\Delta_\lambda H_\lambda^{-1}C^\top P_\lambda^{-1}.
\label{eq:Vclosed}
\end{equation}
For every fixed $\lambda\ge0$,
\begin{equation}
0\preceq V_\lambda\preceq V_0,
\label{eq:Vorder}
\end{equation}
where
\begin{equation}
V_0=(A-CD^{-1}C^\top)^{-1}.
\label{eq:V0}
\end{equation}
\end{theorem}

\begin{corollary}
\label{cor:strict}
Under the conditions of Theorem~\ref{thm:efficiency}, let
$a\in\R^p\setminus\{0\}$. If either
\begin{equation}
a^\top(P_0^{-1}-P_\lambda^{-1})a>0
\label{eq:strict1}
\end{equation}
or
\begin{equation}
\Delta_\lambda^{1/2}H_\lambda^{-1}C^\top P_\lambda^{-1}a\neq0,
\label{eq:strict2}
\end{equation}
then $a^\top V_\lambda a<a^\top V_0a$. If $C=0$, then
$V_\lambda=V_0$ for every $\lambda\ge0$.
\end{corollary}

The auxiliary outcome affects inference for $\beta_0$ through the
cross-information matrix $C$. The matrix $K_\lambda$ increases the
curvature of the survival nuisance equation, while the term involving
$\Delta_\lambda$ accounts for the difference between the derivative and
score covariance matrices. Both contribute to the ordering in
\eqref{eq:Vorder}. When $C=0$, improved estimation of the nuisance
function has no first-order effect on the variance of the target
estimator. The ordering depends on the information identities and
score orthogonality; a general sandwich covariance need not satisfy it.

\subsection{Differences between outcome-specific effects}
\label{sec:local}

\subsubsection{Local differences}

Keep $\beta_0$ fixed and consider a sequence of models in which
$\theta_{Y0,n}$ and $\theta_{T0,n}$ converge to a common value $\theta_0$,
with
\begin{equation}
\theta_{Y0,n}-\theta_{T0,n}=n^{-1/2}h,
\qquad h\in\R^K\text{ fixed}.
\label{eq:localtheta}
\end{equation}
The corresponding function difference is
$n^{-1/2}b_K(x)^\top h$. In this subsection, the information matrices
are evaluated at the common limiting model. We assume that the
population derivative and score covariance matrices along the sequence
converge to those at this limiting model.

\begin{theorem}
\label{thm:local}
Suppose Assumptions~\ref{ass:sampling}--\ref{ass:quad} hold uniformly
along \eqref{eq:localtheta}. For each fixed $\lambda\ge0$,
\begin{equation}
\sqrt n(\widehat\beta_\lambda-\beta_0)
\rightsquigarrow
N\{B_\lambda(h),V_\lambda\},
\label{eq:locallimit}
\end{equation}
where $V_\lambda$ is given by \eqref{eq:Vclosed} and
\begin{equation}
B_\lambda(h)
=-P_\lambda^{-1}CH_\lambda^{-1}K_\lambda h.
\label{eq:Blambda}
\end{equation}
In particular, $B_0(h)=0$.
\end{theorem}

\begin{corollary}
For a scalar target $a^\top\beta_0$, with $a\in\R^p$, the local
asymptotic mean squared error, defined from the limiting distribution
in \eqref{eq:locallimit}, is
\begin{equation}
\operatorname{AMSE}_\lambda(a,h)
=a^\top V_\lambda a+\{a^\top B_\lambda(h)\}^2.
\label{eq:scalaramse}
\end{equation}
The corresponding matrix for the vector parameter is
\begin{equation}
\operatorname{AMSE}_\lambda(h)
=V_\lambda+B_\lambda(h)B_\lambda(h)^\top.
\label{eq:matrixamse}
\end{equation}
\end{corollary}

The mean shift depends on the direction of the difference, not only on
its size. In particular, a direction $h$ satisfying
$CH_\lambda^{-1}K_\lambda h=0$ produces no first-order bias in
$\widehat\beta_\lambda$. Differences of equal magnitude can therefore
have different consequences for inference on $\beta_0$.

\subsubsection{Fixed differences}
\label{sec:fixed}

Now suppose that $\theta_{Y0}-\theta_{T0}\neq0$ is fixed.
Proposition~\ref{prop:fixedpop} shows why a fixed positive penalty can
change the population target. The following result describes the target
as the penalty approaches zero.

\begin{proposition}
\label{prop:vanish}
Suppose the unpenalized population criterion is twice continuously
differentiable near $\vartheta_0$, its Hessian at $\vartheta_0$ is
nonsingular, and $\theta_{Y0}-\theta_{T0}\neq0$ is fixed. Let
$\vartheta_\lambda$ denote the nearby population minimizer under penalty
level $\lambda$. Then, as $\lambda\downarrow0$,
\begin{equation}
\vartheta_\lambda-\vartheta_0=O(\lambda).
\label{eq:poptargetrate}
\end{equation}
If $\lambda_n\to0$ and the estimator is consistent for
$\vartheta_{\lambda_n}$, it is also consistent for $\vartheta_0$.
If, in addition, $\sqrt n\,\lambda_n\to0$ and the stochastic expansion
is uniform for sufficiently small $\lambda$, then
$\widehat\beta_{\lambda_n}$ has the same first-order limiting
distribution as the separate estimator $\widehat\beta_0$.
\end{proposition}

\subsection{Selection of the penalty parameter}
\label{sec:adaptive}

The preceding results suggest retaining a positive penalty when the
two covariate functions agree and reducing it when their difference
persists. Let $\widehat f_Y^{\mathrm{sep}}$ and
$\widehat f_T^{\mathrm{sep}}$ be obtained from separate fits and define
\begin{equation}
\widehat D_n
=\norm{\widehat f_Y^{\mathrm{sep}}-\widehat f_T^{\mathrm{sep}}}_n.
\label{eq:Dhat}
\end{equation}
Under equal covariate effects, suppose
\begin{equation}
\widehat D_n=O_p(r_n),
\qquad r_n\downarrow0.
\label{eq:Drate}
\end{equation}
Choose a deterministic threshold $\tau_n\downarrow0$ satisfying
\begin{equation}
r_n=o(\tau_n).
\label{eq:taurate}
\end{equation}
Let $\omega:[0,\infty)\to[0,1]$ be continuous, equal to one in a
neighborhood of zero, and equal to zero for $u\ge1$. For a prespecified
$\lambda_+>0$, set
\begin{equation}
\widehat\lambda_n
=\lambda_+\,\omega(\widehat D_n/\tau_n).
\label{eq:lambdahat}
\end{equation}

\begin{theorem}
\label{thm:adaptive}
Assume that the separate pilot estimators are consistent in $L_2(P_X)$
and also satisfy
\[
\norm{\widehat f_Y^{\mathrm{sep}}-f_{Y0}}_n
+
\norm{\widehat f_T^{\mathrm{sep}}-f_{T0}}_n
=o_p(1).
\]
Suppose \eqref{eq:Drate}--\eqref{eq:taurate} hold under equality and
that the fixed-difference function is square-integrable. Then:
\begin{enumerate}[label=(\roman*)]
\item If $f_{Y0}=f_{T0}$, then
$\widehat\lambda_n\to_p\lambda_+$.
\item If $D_0=\norm{f_{Y0}-f_{T0}}_{P_X}>0$ is fixed, then
$\widehat\lambda_n\to_p0$.
\item If the first-order expansion of $\widehat\beta_\lambda$ is
stochastically equicontinuous in $\lambda$ near $\lambda_+$ and zero,
then
\[
\sqrt n\bigl(
\widehat\beta_{\widehat\lambda_n}-\widehat\beta_{\lambda_+}
\bigr)=o_p(1)
\]
under equality, and
\[
\sqrt n\bigl(
\widehat\beta_{\widehat\lambda_n}-\widehat\beta_0
\bigr)=o_p(1)
\]
under a fixed difference.
\end{enumerate}
\end{theorem}

Along the local sequence \eqref{eq:localtheta}, suppose
$\widehat D_n=O_p(n^{-1/2}+r_n)$. If
$(n^{-1/2}+r_n)/\tau_n\to0$, then
$\widehat\lambda_n\to_p\lambda_+$. The procedure therefore retains the
local mean shift in Theorem~\ref{thm:local}. The conclusions of
Theorem~\ref{thm:adaptive} are pointwise statements under equality and
fixed differences, not a guarantee of uniform validity over shrinking
alternatives.

For comparison, the penalty minimizing the local asymptotic mean
squared error for a specified target $a^\top\beta_0$ satisfies
\begin{equation}
\lambda^*(a,h)
\in
\argmin_{\lambda\in\Lambda}
\left[
a^\top V_\lambda a+\{a^\top B_\lambda(h)\}^2
\right],
\label{eq:oracle}
\end{equation}
where $\Lambda\subset[0,\infty)$ is a compact set containing zero.
Because $h$ is unknown, \eqref{eq:oracle} is a theoretical benchmark
rather than an implementable selection rule. Predictive
cross-validation addresses a different objective and need not select
the penalty most suitable for inference on $\beta_0$.

\subsection{Growing-sieve extension}
\label{sec:growing}

Let $K=K_n\to\infty$ in \eqref{eq:sieve}. Write
$A_n,C_n,D_n,E_n$ for the population information blocks at dimension
$K_n$, and define
\[
G^{(n)}=\E\{b_{K_n}(X)b_{K_n}(X)^\top\}.
\]
The empirical Gram matrix is still denoted by $G_n$. For local
alternatives, the information blocks are evaluated at the common
equal-effect limit. Subscripts $n$ on $H_{\lambda,n}$, $P_{\lambda,n}$,
$K_{\lambda,n}$, and $V_{\lambda,n}$ indicate that the fixed-sieve
formulas are evaluated using these population blocks. The following
conditions allow the finite-dimensional covariance calculation to pass
to the limit.

\begin{assumption}[Sieve approximation]
\label{ass:sieveapprox}
There are sieve approximations $f_{Y,K_n}$ and $f_{T,K_n}$ such that
\[
\norm{f_{Y,K_n}-f_{Y0}}_{P_X}
+
\norm{f_{T,K_n}-f_{T0}}_{P_X}
=O(a_{K_n}),
\qquad a_{K_n}\to0.
\]
Along local alternatives, the same condition holds for the corresponding
sequence of true functions.
\end{assumption}

\begin{assumption}[Nuisance estimation and profile remainder]
\label{ass:nuisancerate}
For each fixed $\lambda\in\Lambda$, under equality and the local
alternatives considered below, the joint nuisance estimation error is
$O_p(r_n)$ for a sequence $r_n\to0$. The remainder in the profiled
estimating equation for $\beta$ is $o_p(n^{-1/2})$, including the
contribution of sieve approximation. When the estimation remainder is
quadratic in the nuisance error, the rate
\[
r_n=o(n^{-1/4})
\]
is sufficient for that part of the remainder, provided that the
approximation bias in the profiled target equation is also negligible
at the root-$n$ scale.
\end{assumption}

\begin{assumption}[Hessian and score approximation]
\label{ass:uniformblocks}
Uniformly over $\lambda\in\Lambda$:
\begin{enumerate}[label=(\alph*)]
\item the empirical block Hessians converge in operator norm to their
population sieve counterparts;
\item the score for $\beta$, after profiling the nuisance coordinates,
admits an asymptotically linear representation satisfying a central
limit theorem;
\item the inverses and Schur complements in
Theorems~\ref{thm:efficiency} and \ref{thm:local} are uniformly bounded
on the subspaces relevant to $\beta$;
\item replacing $G_n$ by $G^{(n)}$ contributes an $o_p(1)$ term to the
root-$n$ expansion of the target estimator.
\end{enumerate}
\end{assumption}

\begin{theorem}
\label{thm:growing}
Suppose Assumptions~\ref{ass:sampling} and
\ref{ass:sieveapprox}--\ref{ass:uniformblocks} hold, together with
sieve versions of Assumptions~\ref{ass:pd}--\ref{ass:crossorth}.
Suppose also that $V_{\lambda,n}\to V_\lambda$ for each fixed
$\lambda\in\Lambda$. Then:
\begin{enumerate}[label=(\roman*)]
\item Under $f_{Y0}=f_{T0}$,
\[
\sqrt n(\widehat\beta_\lambda-\beta_0)
\rightsquigarrow N(0,V_\lambda).
\]
\item If $V_{\lambda,n}\preceq V_{0,n}$ for every $n$, then
$V_\lambda\preceq V_0$.
\item For a local difference represented in the sieve by
$n^{-1/2}h_n$, with $h_n\in\R^{K_n}$, suppose any representation error
is negligible in the profiled root-$n$ expansion and
\[
B_{\lambda,n}(h_n)
=
-P_{\lambda,n}^{-1}C_nH_{\lambda,n}^{-1}K_{\lambda,n}h_n
\]
converges to a finite vector $B_\lambda$. Then the local limiting
distribution is $N(B_\lambda,V_\lambda)$.
\end{enumerate}
\end{theorem}

\section{Simulation study}
\label{sec:simulation}

\subsection{Design}
We used simulations to evaluate the finite-sample performance of the proposed estimators.  We generated a scalar baseline covariate $X\sim\operatorname{Unif}(-1,1)$ and an independent randomized treatment indicator $Z\sim\operatorname{Bernoulli}(1/2)$. Event times followed
\[
\lambda(t\mid Z,X)=0.08\exp\{\beta_0Z+f_{T0}(X)\},\qquad \beta_0=\log(0.65),
\]
and the auxiliary outcome followed
\[
Y=0.4Z+f_{Y0}(X)+\varepsilon,\qquad \varepsilon\sim N(0,\sigma_Y^2).
\]
Independent exponential censoring was calibrated to give either 30\% or 60\% marginal censoring. The function $f_{T0}$ was represented by a centered cubic B-spline basis with seven linearly independent columns. Its coefficient vector was obtained from the spline projection of $0.95\sin(\pi x)+0.55(x^2-1/3)$ and rescaled so that $\|f_{T0}\|_{L_2(P_X)}\approx0.9$. Outcome differences were introduced through
\[
f_{Y0}(x)=f_{T0}(x)+\delta g(x),\qquad \|g\|_{L_2(P_X)}=1.
\]
For each discrepancy function $g$, let $h_g\in\R^K$ denote its coefficient vector, so that $g(x)=b_K(x)^\top h_g$. To study the role of direction as well as magnitude, one coefficient direction was selected to maximize $|CH_1^{-1}K_1h|$ subject to $h^\top Gh=1$, using high-accuracy Monte Carlo approximations to the population fixed-sieve matrices. A second direction was normalized to satisfy $CH_1^{-1}K_1h=0$ up to numerical precision. Thus two outcome differences can have the same $L_2(P_X)$ magnitude but different first-order effects on the treatment coefficient.

We compared separate estimation, fixed penalties $\lambda\in\{0.25,1,4\}$, complete pooling with $f_Y=f_T$, and the adaptive penalty from Section~\ref{sec:adaptive}. For the adaptive procedure we used $\lambda_+=1$, $\tau_n=1.75n^{-1/4}$, and
\[
\omega(u)=
\begin{cases}
1,&u\le1/2,\\
2(1-u),&1/2<u<1,\\
0,&u\ge1.
\end{cases}
\]
Standard errors were computed from the plug-in sandwich covariance corresponding to the derivative matrix of the penalized estimating equations and the outcome score covariance blocks. The data-generating mechanism satisfies the score-orthogonality condition because the auxiliary error is independent of the event and censoring processes conditional on $(Z,X)$. Each setting used 1000 Monte Carlo replications. We report bias, empirical standard deviation (ESD), average estimated standard error (ASE), empirical coverage of nominal 95\% Wald intervals, root mean squared error (RMSE), and relative efficiency compared with separate estimation.

\subsection{Equal baseline covariate effects}
Table~\ref{tab:common} considers $f_{Y0}=f_{T0}$ with $\sigma_Y=1$. Empirical coverage is close to 95\% for all procedures. Penalized estimation reduces the dispersion of $\widehat\beta$ in every setting, as predicted by Theorem~\ref{thm:efficiency}. The gain is modest with 30\% censoring at $n=600$, but is more visible when survival information is reduced. With 60\% censoring and $n=300$, for example, the relative efficiencies are 1.071 for $\lambda=1$ and 1.081 for complete pooling. The adaptive estimator retains most of the gain without imposing equality in advance.

\begin{table}[ht]
\centering
\caption{Equal baseline covariate effects across outcomes. Bias is multiplied by 100; Cov. is empirical coverage (\%). Relative efficiency is the empirical variance under separate estimation divided by the empirical variance of the indicated estimator.}
\label{tab:common}
\small
\begin{tabular}{rrlrrrrr}
\toprule
$n$ & Cens. (\%) & Method & Bias$\times100$ & ESD & ASE & Cov. & Rel. eff.\\
\midrule
300&30&Separate&-1.50&0.147&0.145&95.1&1.000\\
300&30&Penalty $\lambda=1$&-1.05&0.144&0.142&95.5&1.050\\
300&30&Adaptive&-1.14&0.144&0.143&95.3&1.045\\
300&30&Complete pooling&-0.93&0.143&0.142&95.6&1.058\\
600&30&Separate&0.13&0.100&0.101&95.3&1.000\\
600&30&Penalty $\lambda=1$&0.37&0.099&0.100&95.3&1.016\\
600&30&Adaptive&0.36&0.099&0.100&95.3&1.016\\
600&30&Complete pooling&0.44&0.099&0.099&95.0&1.016\\
300&60&Separate&-1.74&0.201&0.193&94.3&1.000\\
300&60&Penalty $\lambda=1$&-0.90&0.194&0.189&95.0&1.071\\
300&60&Adaptive&-1.35&0.197&0.190&94.8&1.041\\
300&60&Complete pooling&-0.71&0.193&0.189&95.0&1.081\\
600&60&Separate&-0.31&0.134&0.133&94.5&1.000\\
600&60&Penalty $\lambda=1$&0.05&0.132&0.132&95.0&1.031\\
600&60&Adaptive&-0.03&0.133&0.132&94.6&1.022\\
600&60&Complete pooling&0.13&0.132&0.132&94.9&1.033\\
\bottomrule
\end{tabular}
\end{table}

Table~\ref{tab:auxnoise} shows the effect of auxiliary-outcome precision using common random numbers across the three values of $\sigma_Y$. A more precise auxiliary outcome gives a larger efficiency gain. For complete pooling, relative efficiency decreases from 1.052 at $\sigma_Y=0.5$ to 1.019 at $\sigma_Y=2$. The same pattern is present for the fixed and adaptive penalties.

\begin{table}[ht]
\centering
\caption{Sensitivity to the auxiliary noise level under equal baseline covariate effects ($n=600$, 30\% censoring). Bias is multiplied by 100.}
\label{tab:auxnoise}
\small
\begin{tabular}{rlrrrrr}
\toprule
$\sigma_Y$ & Method & Bias$\times100$ & ESD & ASE & Cov. & Rel. eff.\\
\midrule
0.5&Separate&-0.57&0.099&0.101&95.2&1.000\\
0.5&Penalty $\lambda=1$&-0.21&0.097&0.099&95.1&1.040\\
0.5&Adaptive&-0.23&0.097&0.099&95.1&1.039\\
0.5&Complete pooling&-0.07&0.096&0.099&95.3&1.052\\
1&Separate&-0.57&0.099&0.101&95.2&1.000\\
1&Penalty $\lambda=1$&-0.25&0.097&0.100&95.2&1.032\\
1&Adaptive&-0.28&0.097&0.100&95.2&1.030\\
1&Complete pooling&-0.16&0.097&0.099&95.1&1.038\\
2&Separate&-0.57&0.099&0.101&95.2&1.000\\
2&Penalty $\lambda=1$&-0.35&0.098&0.100&95.3&1.018\\
2&Adaptive&-0.43&0.098&0.100&95.3&1.013\\
2&Complete pooling&-0.32&0.098&0.100&95.3&1.019\\
\bottomrule
\end{tabular}
\end{table}

\subsection{Fixed and local outcome differences}
We next fixed $n=600$, 30\% censoring, and $\sigma_Y=1$, and increased the magnitude $\delta$ in the direction with a strong first-order effect on the treatment coefficient. Table~\ref{tab:fixed} and Figures~\ref{fig:rmse}--\ref{fig:coverage} illustrate the effects of a fixed difference, consistent with the population-target change described in Proposition~\ref{prop:fixedpop}. When $\delta=0$, stronger penalization reduces variance. With a persistent difference, a fixed positive penalty introduces bias, and the bias increases with both $\delta$ and $\lambda$. At $\delta=0.30$, for example, $\lambda=1$ has bias $-0.027$ and 91.6\% coverage, whereas the adaptive estimator has bias $-0.010$ and an average penalty of 0.20. At $\delta=0.60$, the adaptive penalty is essentially zero and its performance is indistinguishable from separate estimation. Complete pooling gives the largest gain at $\delta=0$ but deteriorates most rapidly when the equality restriction is false.

\begin{table}[ht]
\centering
\caption{Fixed outcome differences in the direction that affects the treatment estimate, $f_{Y0}-f_{T0}=\delta g$ ($n=600$, 30\% censoring). Bias is multiplied by 100.}
\label{tab:fixed}
\small
\resizebox{\textwidth}{!}{%
\begin{tabular}{rlrrrrr}
\toprule
$\delta$ & Method & Bias$\times100$ & ESD & Cov. & RMSE & Mean $\lambda$\\
\midrule
0.00&Separate&-0.49&0.099&95.9&0.099&0.00\\
0.00&Penalty $\lambda=0.25$&-0.34&0.098&96.1&0.098&0.25\\
0.00&Penalty $\lambda=1$&-0.24&0.097&95.8&0.097&1.00\\
0.00&Penalty $\lambda=4$&-0.19&0.097&96.0&0.097&4.00\\
0.00&Adaptive&-0.26&0.097&95.8&0.097&0.90\\
0.00&Complete pooling&-0.17&0.097&95.9&0.097&--\\
0.15&Separate&-0.35&0.100&95.5&0.100&0.00\\
0.15&Penalty $\lambda=0.25$&-0.97&0.100&95.8&0.101&0.25\\
0.15&Penalty $\lambda=1$&-1.42&0.101&95.9&0.102&1.00\\
0.15&Penalty $\lambda=4$&-1.66&0.101&95.3&0.103&4.00\\
0.15&Adaptive&-1.22&0.101&95.9&0.102&0.73\\
0.15&Complete pooling&-1.76&0.102&95.3&0.103&--\\
0.30&Separate&-0.35&0.100&94.7&0.100&0.00\\
0.30&Penalty $\lambda=0.25$&-1.73&0.103&93.3&0.104&0.25\\
0.30&Penalty $\lambda=1$&-2.71&0.105&91.6&0.109&1.00\\
0.30&Penalty $\lambda=4$&-3.21&0.107&90.8&0.111&4.00\\
0.30&Adaptive&-1.02&0.102&93.9&0.102&0.20\\
0.30&Complete pooling&-3.44&0.107&90.7&0.113&--\\
0.60&Separate&-0.28&0.104&94.2&0.104&0.00\\
0.60&Penalty $\lambda=0.25$&-3.24&0.112&90.1&0.117&0.25\\
0.60&Penalty $\lambda=1$&-5.35&0.119&87.1&0.130&1.00\\
0.60&Penalty $\lambda=4$&-6.41&0.123&84.8&0.138&4.00\\
0.60&Adaptive&-0.28&0.104&94.2&0.104&0.00\\
0.60&Complete pooling&-6.86&0.124&84.1&0.142&--\\
\bottomrule
\end{tabular}}
\end{table}

\begin{figure}[ht]
\centering
\includegraphics[width=0.68\textwidth]{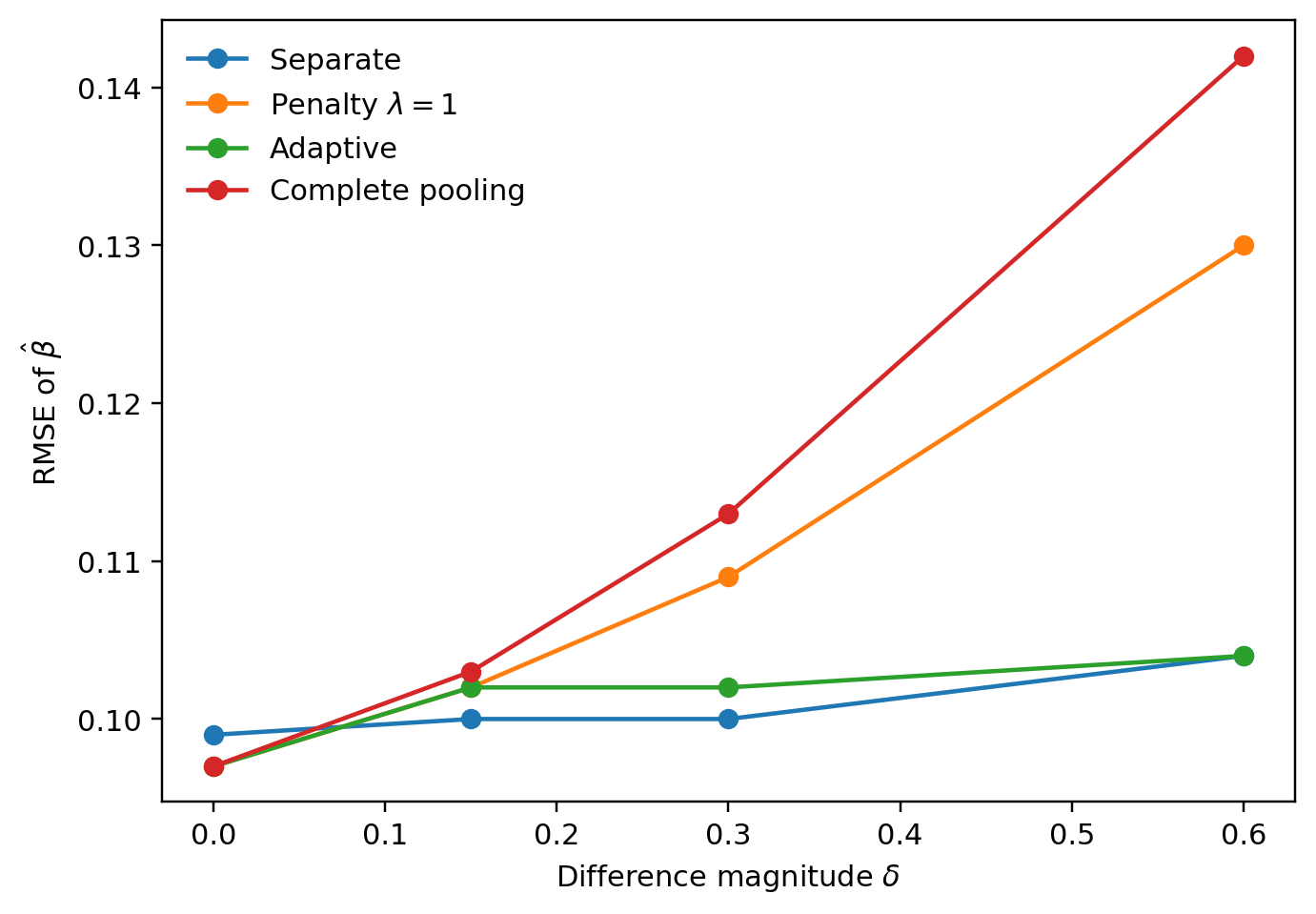}
\caption{RMSE of the treatment log-hazard-ratio estimator as the fixed outcome difference increases.}
\label{fig:rmse}
\end{figure}

\begin{figure}[ht]
\centering
\includegraphics[width=0.68\textwidth]{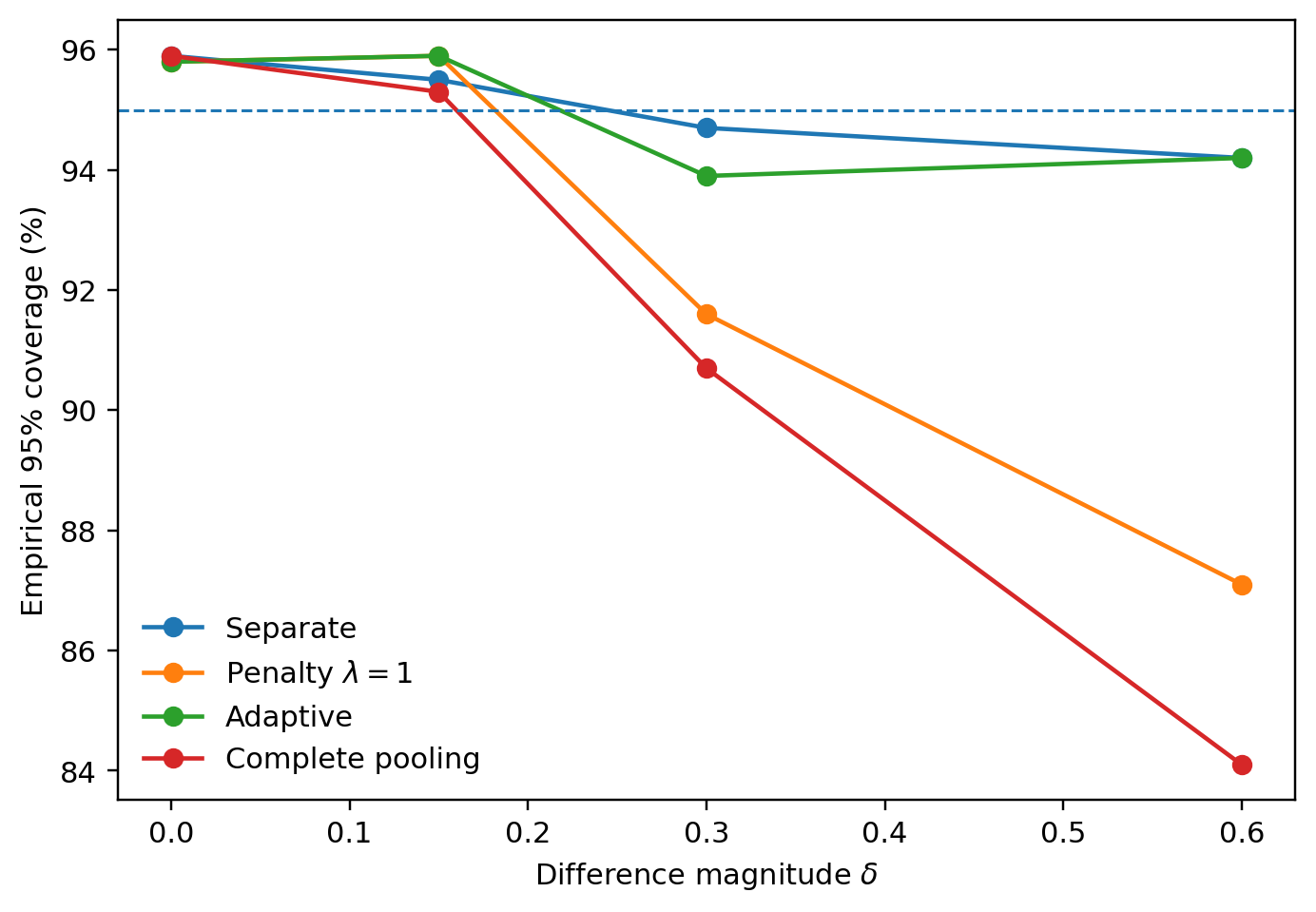}
\caption{Empirical coverage of nominal 95\% Wald intervals as the fixed outcome difference increases.}
\label{fig:coverage}
\end{figure}

To examine Theorem~\ref{thm:local} directly, we generated $f_{Y0,n}-f_{T0,n}=4n^{-1/2}g$ in the same direction. To compare with the separate estimator, Table~\ref{tab:local} and
Figure~\ref{fig:localbias} report the incremental scaled bias
\[
\sqrt n\{\operatorname{Bias}(\widehat\beta_\lambda)-\operatorname{Bias}(\widehat\beta_0)\},
\]
which measures the additional bias relative to separate estimation. The empirical values move toward the population mean shifts from Theorem~\ref{thm:local}. For $\lambda=1$, for example, the empirical scaled bias is $-0.252$, $-0.283$, and $-0.300$ for $n=300,600,1200$, compared with the theoretical value $-0.327$. The pattern is consistent with the root-$n$ local expansion in Theorem~\ref{thm:local}.

\begin{table}[ht]
\centering
\caption{Local difference $f_{Y0,n}-f_{T0,n}=4n^{-1/2}g$. Empirical shift is $\sqrt n\{\operatorname{Bias}(\widehat\beta_\lambda)-\operatorname{Bias}(\widehat\beta_0)\}$; Theory is $B_\lambda(4h_g)$, where $g(x)=b_K(x)^\top h_g$.}
\label{tab:local}
\small
\begin{tabular}{rrrrrr}
\toprule
$n$ & $\lambda$ & Empirical shift & Theory & ESD & Cov.\\
\midrule
300&0.25&-0.148&-0.190&0.143&94.6\\
300&1&-0.252&-0.327&0.145&94.3\\
300&4&-0.304&-0.399&0.146&93.6\\
600&0.25&-0.164&-0.190&0.102&95.1\\
600&1&-0.283&-0.327&0.103&94.2\\
600&4&-0.345&-0.399&0.103&93.8\\
1200&0.25&-0.174&-0.190&0.071&95.0\\
1200&1&-0.300&-0.327&0.071&94.3\\
1200&4&-0.367&-0.399&0.072&93.8\\
\bottomrule
\end{tabular}
\end{table}

\begin{figure}[ht]
\centering
\includegraphics[width=0.68\textwidth]{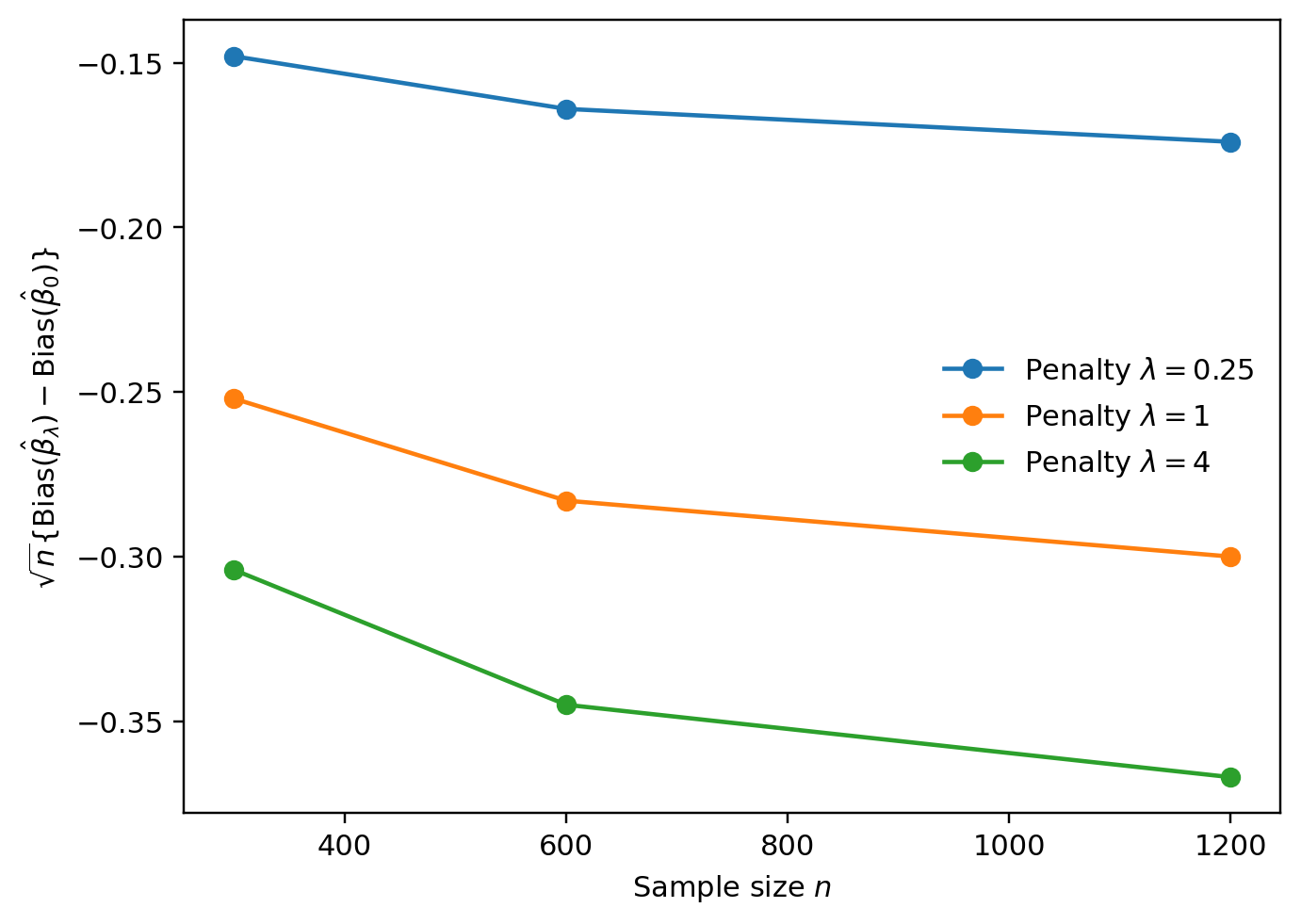}
\caption{Incremental scaled bias under the local difference. The limiting values from Theorem~\ref{thm:local} are $-0.190$, $-0.327$, and $-0.399$ for $\lambda=0.25,1,4$, respectively.}
\label{fig:localbias}
\end{figure}

\subsection{Direction of the outcome difference}
Finally, Table~\ref{tab:alignment} compares two outcome differences with the same $L_2(P_X)$ magnitude, $\delta=0.40$, but different first-order effects on the treatment parameter. The survival data are paired across the two designs; only the direction of $f_{Y0}-f_{T0}$ changes. With the treatment-relevant difference, $\lambda=1$ adds an incremental bias of about $-0.032$ and reduces coverage to 92.1\%. With the approximately orthogonal difference, the incremental bias is only 0.0029 and coverage is essentially unchanged. Thus the $L_2$ distance between the two functions does not by itself determine the effect of penalization on $\widehat\beta$.

The adaptive rule uses only $\|\widehat f_Y^{\mathrm{sep}}-\widehat f_T^{\mathrm{sep}}\|_n$ and is therefore conservative in this comparison: it largely turns off penalization for both directions even though the approximately orthogonal difference has little effect on the treatment estimate. This illustrates why the target-specific AMSE criterion in Section~\ref{sec:adaptive} is useful as a theoretical benchmark.

\begin{table}[ht]
\centering
\caption{Outcome differences with the same $L_2(P_X)$ magnitude ($\delta=0.40$) but different first-order effects on the treatment parameter. Incremental bias is relative to separate estimation and is multiplied by 100.}
\label{tab:alignment}
\small
\resizebox{\textwidth}{!}{%
\begin{tabular}{llrrrrrr}
\toprule
Direction & Method & Bias$\times100$ & Inc. bias$\times100$ & ESD & Cov. & RMSE & Mean $\lambda$\\
\midrule
Treatment-relevant&Separate&0.08&0.00&0.101&94.2&0.101&0.00\\
Treatment-relevant&Penalty $\lambda=1$&-3.08&-3.17&0.109&92.1&0.113&1.00\\
Treatment-relevant&Adaptive&-0.11&-0.19&0.102&94.2&0.102&0.04\\
Treatment-relevant&Complete pooling&-4.06&-4.14&0.112&91.4&0.119&--\\
Approximately orthogonal&Separate&0.08&0.00&0.101&94.2&0.101&0.00\\
Approximately orthogonal&Penalty $\lambda=1$&0.37&0.29&0.100&94.3&0.100&1.00\\
Approximately orthogonal&Adaptive&0.08&-0.00&0.101&94.1&0.101&0.02\\
Approximately orthogonal&Complete pooling&0.57&0.49&0.101&93.6&0.102&--\\
\bottomrule
\end{tabular}}
\end{table}

Overall, the simulations agree with the theoretical picture. Information sharing improves precision when the two baseline covariate functions agree, with larger gains when the auxiliary outcome is more informative or the survival outcome is more heavily censored. A persistent difference can convert the variance reduction into bias, while the adaptive procedure returns toward separate estimation. Under local differences, the mean shift occurs on the predicted root-$n$ scale and depends strongly on the direction of the difference rather than only its norm.

\FloatBarrier
\section{ACTG 175 data analysis}
\label{sec:application}

\subsection{Data and analysis}
We illustrate the procedure using AIDS Clinical Trials Group Study 175, a randomized trial comparing nucleoside antiretroviral regimens in HIV-infected adults with baseline CD4 counts between 200 and 500 cells/mm$^3$ \citep{hammer1996}. We use the public ACTG 175 data set distributed with the BART R package, which contains 2,139 participants. The survival outcome is the right-censored time to the first occurrence of a substantial CD4 decline, progression to AIDS, or death. In this data set, 521 participants experienced the composite event (24.4\%). We define $Z=1$ for assignment to any of the three active regimens other than zidovudine monotherapy and $Z=0$ for zidovudine monotherapy. There were 1,607 and 532 participants in these two groups, respectively.

The auxiliary outcome was the change in CD4 count from baseline to 20 weeks,
\[
Y_i=\operatorname{CD4}_{i,20}-\operatorname{CD4}_{i,0}.
\]
The mean change was $-17.1$ cells/mm$^3$ in the zidovudine-monotherapy group and 33.3 cells/mm$^3$ in the other-treatment group. Before fitting the joint model, the auxiliary outcome was centered and divided by its sample standard deviation, 122.3 cells/mm$^3$, to express the auxiliary covariate function on the sample-standardized outcome scale. A separate partially linear analysis estimated an adjusted treatment difference of 50.3 cells/mm$^3$ in 20-week CD4 change, with robust standard error 5.1 cells/mm$^3$.

The baseline covariates $X$ were age, body weight, Karnofsky score, duration of prior antiretroviral therapy, baseline CD4 count, baseline CD8 count, hemophilia, homosexual activity, history of intravenous drug use, prior non-zidovudine antiretroviral therapy, zidovudine use in the preceding 30 days, race, sex, antiretroviral-treatment history, and symptomatic status. The six continuous variables were represented by cubic B-splines with four degrees of freedom; binary variables entered linearly. After centering and removing linearly redundant columns, the sieve dimension was 31. The Cox partial likelihood used the Breslow approximation for tied event times. Standard errors in this section use the full empirical sandwich covariance. In particular, the meat matrix retains the empirical covariance between the auxiliary residual score and the survival martingale score, as well as the subject-level contribution of the empirical quadratic penalty. 

\subsection{Treatment-effect estimates over the penalty path}
We fitted the model over
\[
\lambda\in\{0,0.03,0.1,0.3,1,3,10,30\},
\]
and also fitted the complete-pooling limit. Table~\ref{tab:actg} gives representative values. With separate estimation, the estimated treatment log-hazard ratio was $-0.689$ (robust SE 0.096), corresponding to a hazard ratio of 0.502 with 95\% confidence interval (0.416, 0.606). Moderate penalization reduced the standard error slightly. At $\lambda=0.3$, for example, the standard error was 0.092, corresponding to a variance ratio of approximately 1.084 relative to separate estimation, while the point estimate moved from $-0.689$ to $-0.667$. Larger penalties produced little additional change in precision and moved the estimate further toward the complete-pooling value. Figures~\ref{fig:actgpath} and \ref{fig:actgse} display the treatment estimate and standard error over the full penalty path. Because the population target can depend on the penalty when the covariate functions differ, a smaller standard error alone does not establish improved inference for the original survival coefficient.

\begin{table}[ht]
\centering
\caption{ACTG 175 treatment-effect estimates over the penalty path. Standard errors use the full empirical sandwich covariance. $\widehat D_\lambda=\|\widehat f_{Y,\lambda}-\widehat f_{T,\lambda}\|_n$.}
\label{tab:actg}
\small
\resizebox{\textwidth}{!}{%
\begin{tabular}{lrrrrlr}
\toprule
Method & $\lambda$ & $\widehat\beta$ & Robust SE & Hazard ratio & 95\% CI & $\widehat D_\lambda$\\
\midrule
Separate&0&-0.689&0.096&0.502&(0.416, 0.606)&0.795\\
Penalized&0.10&-0.675&0.093&0.509&(0.424, 0.611)&0.484\\
Penalized&0.30&-0.667&0.092&0.513&(0.428, 0.616)&0.289\\
Penalized&1.00&-0.660&0.093&0.517&(0.431, 0.620)&0.120\\
Penalized&3.00&-0.658&0.093&0.518&(0.431, 0.622)&0.045\\
Complete pooling&$\infty$&-0.657&0.094&0.519&(0.432, 0.623)&0\\
Adaptive&0&-0.689&0.096&0.502&(0.416, 0.606)&0.795\\
\bottomrule
\end{tabular}}
\end{table}

\begin{figure}[ht]
\centering
\includegraphics[width=0.72\textwidth]{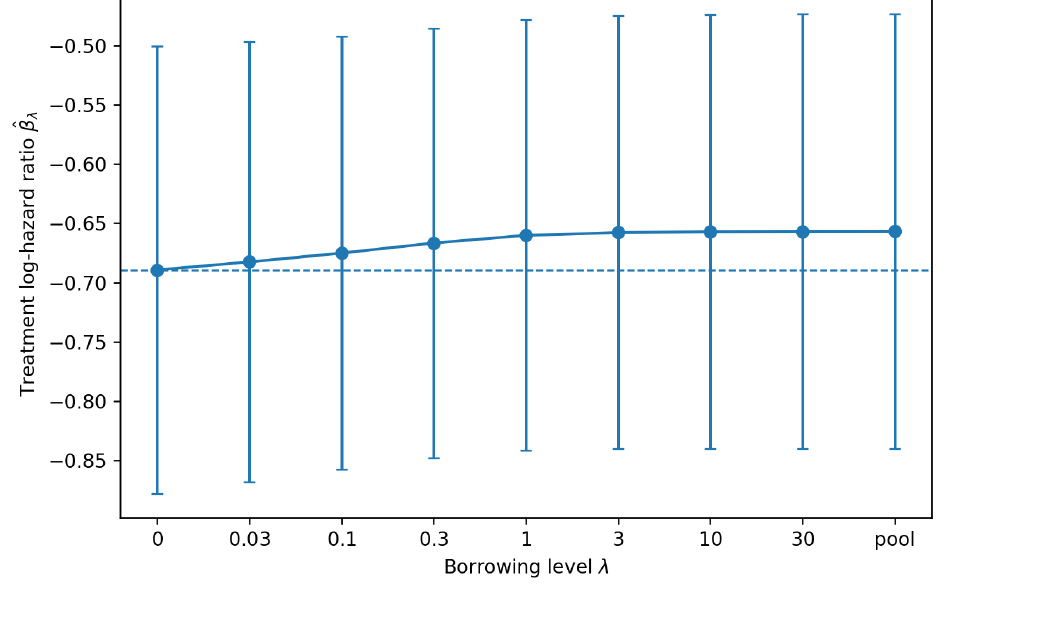}
\caption{Estimated treatment log-hazard ratio and robust 95\% confidence interval over the penalty path in ACTG 175. The dashed horizontal line is the separate-estimation value. The final point is the complete-pooling fit.}
\label{fig:actgpath}
\end{figure}

\begin{figure}[ht]
\centering
\includegraphics[width=0.68\textwidth]{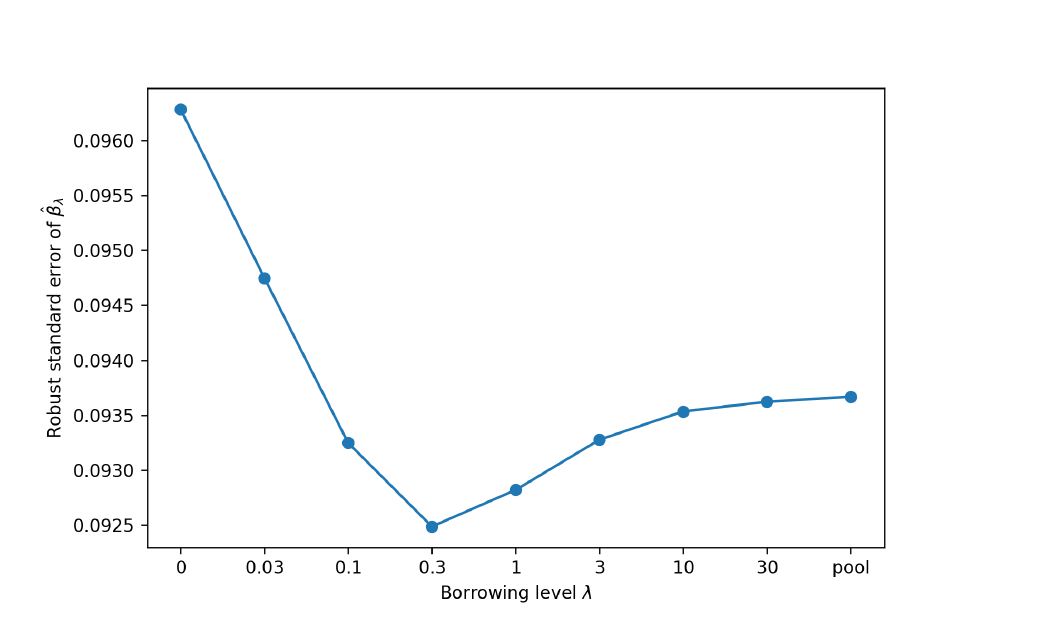}
\caption{Robust standard error of the treatment log-hazard-ratio estimator over the penalty path. The estimated standard error is modestly smaller for intermediate values of $\lambda$.}
\label{fig:actgse}
\end{figure}

The separate outcome fits produced
\[
\widehat D_n=\|\widehat f_Y^{\mathrm{sep}}-\widehat f_T^{\mathrm{sep}}\|_n=0.795.
\]
For comparability with the simulation study, we used the same prespecified adaptive rule, with $\lambda_+=1$ and $\tau_n=1.75n^{-1/4}$. Here $\tau_n=0.257$, so $\widehat D_n/\tau_n=3.09$ and the rule selects $\widehat\lambda=0$. The adaptive estimate is therefore the separate-estimation result. This is useful in the present data: the auxiliary outcome is strongly associated with randomized treatment, but its estimated baseline covariate effects are not sufficiently close to those of the survival outcome to support penalization under the prespecified rule.

\subsection{Sensitivity analyses}
We examined two sensitivity analyses. First, we repeated the analysis using three, four, or five spline degrees of freedom for each continuous baseline variable. In all three cases the adaptive rule selected $\widehat\lambda=0$. The separate-estimation hazard-ratio estimates were 0.509, 0.502, and 0.497, respectively, while the corresponding estimates at $\lambda=1$ were 0.518, 0.517, and 0.516. The qualitative conclusion was therefore insensitive to moderate changes in sieve complexity.

Second, because the penalty depends on the relative scale of the outcome-specific covariate functions, we repeated the analysis after multiplying the standardized CD4-change outcome by 0.5 and 2. Table~\ref{tab:scale} shows that fixed-penalty estimates are scale dependent, as expected. For example, the complete-pooling estimate varies from $-0.645$ to $-0.691$ across these normalizations. In contrast, the adaptive procedure selected $\widehat\lambda=0$ under all three normalizations because the pilot differences were 0.764, 0.795, and 0.987, all well above the threshold 0.257. Figure~\ref{fig:scale} shows the corresponding penalty paths.

\begin{table}[ht]
\centering
\caption{Sensitivity to the scale used for the auxiliary outcome. The scale multiplier is applied after centering and standardizing 20-week CD4 change.}
\label{tab:scale}
\small
\begin{tabular}{r l rrr}
\toprule
Scale multiplier & Method & $\widehat\beta$ & Robust SE & Hazard ratio\\
\midrule
0.5&$\lambda=1$&-0.652&0.092&0.521\\
0.5&Complete pooling&-0.645&0.092&0.525\\
0.5&Adaptive ($\widehat\lambda=0$)&-0.689&0.096&0.502\\
1.0&$\lambda=1$&-0.660&0.093&0.517\\
1.0&Complete pooling&-0.657&0.094&0.519\\
1.0&Adaptive ($\widehat\lambda=0$)&-0.689&0.096&0.502\\
2.0&$\lambda=1$&-0.680&0.096&0.507\\
2.0&Complete pooling&-0.691&0.100&0.501\\
2.0&Adaptive ($\widehat\lambda=0$)&-0.689&0.096&0.502\\
\bottomrule
\end{tabular}
\end{table}

\begin{figure}[ht]
\centering
\includegraphics[width=0.62\textwidth]{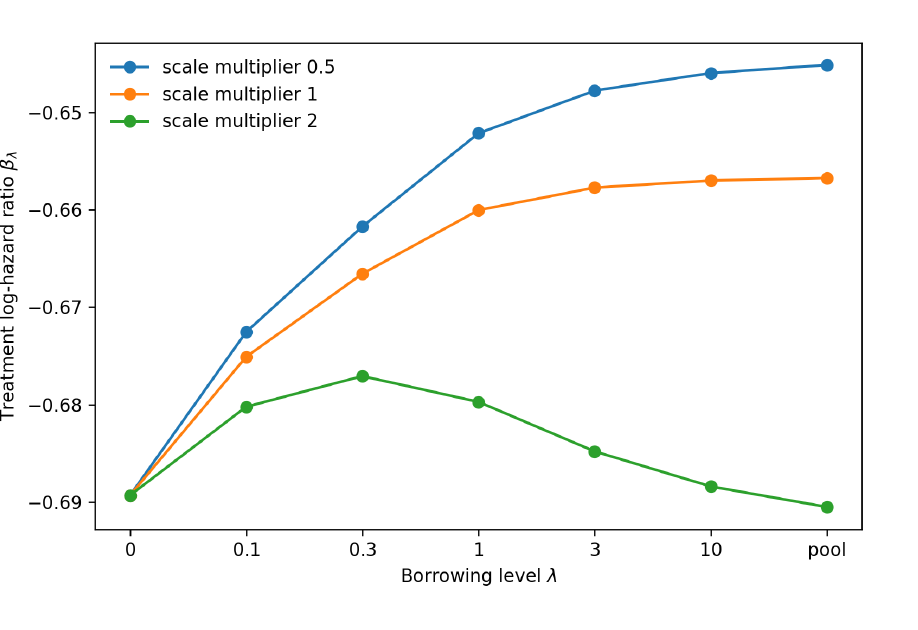}
\caption{Sensitivity of the fixed-penalty treatment estimate to normalization of the auxiliary outcome. The adaptive rule selects separate estimation under all three normalizations.}
\label{fig:scale}
\end{figure}
\FloatBarrier

\clearpage
\appendix
\section{Proofs of the main results}
\label{app:proofs}

For readability, the main text states the theoretical results and their interpretation. Detailed proofs are collected here in the order in which the results appear.

\subsection*{Proof of Proposition~\ref{prop:exacttarget}}
\begin{proof}
For every admissible $(\alpha,\beta,f_Y,f_T)$, separate optimality gives
\[
L_Y(\alpha,f_Y)\ge L_Y(\alpha_0,f_0),
\qquad
L_T(\beta,f_T)\ge L_T(\beta_0,f_0).
\]
The discrepancy term is nonnegative. When the functions coincide, it equals
\[
\frac\lambda2\norm{f_0-f_0}_{P_X}^2=0.
\]
Therefore
\begin{align*}
Q(\alpha,\beta,f_Y,f_T;\lambda)
&\ge L_Y(\alpha_0,f_0)+L_T(\beta_0,f_0)\\
&=Q(\alpha_0,\beta_0,f_0,f_0;\lambda),
\end{align*}
which proves global optimality. If the outcome-specific minimizers are
unique, equality requires both loss differences to be zero. Hence
$(\alpha,f_Y)=(\alpha_0,f_0)$ and
$(\beta,f_T)=(\beta_0,f_0)$, proving uniqueness.
\end{proof}

\subsection*{Proof of Proposition~\ref{prop:fixedpop}}
\begin{proof}
At the interior minimizers of the two differentiable outcome criteria,
their directional derivatives vanish along every admissible path. Thus,
for an admissible direction $(h_Y,h_T)$,
\begin{align*}
&\left.\frac{d}{dt}
Q(\alpha_0,\beta_0,f_{Y0}+t h_Y,f_{T0}+t h_T;\lambda)
\right|_{t=0}\\
&\qquad=\lambda\ip{d_0}{h_Y-h_T}_{P_X},
\end{align*}
by \eqref{eq:penalty-direction}. The assumed nonzero inner product gives
a nonzero derivative of the full criterion, so the unpenalized truth is
not a stationary point.
\end{proof}

\subsection*{Proof of Lemma~\ref{lem:orth}}
\begin{proof}
Write $q(Z,X)=\widetilde b_K(X,Z)/\sigma_Y$, so that $s_Y=\varepsilon q(Z,X)$. The Cox score can be written as
\[
s_S=\int_0^\tau a(t,Z,X)\,dM(t),
\qquad
a(t,Z,X)=W-\bar W(t).
\]
Because $q$ and $a$ are predictable functions of baseline variables and time,
\begin{align*}
\E(s_Ys_S^\top)
&=\E\left[
\varepsilon q(Z,X)
\left\{\int_0^\tau a(t,Z,X)\,dM(t)\right\}^\top
\right]\\
&=\E\left[
q(Z,X)
\int_0^\tau a(t,Z,X)^\top
\E\{\varepsilon\,dM(t)\mid Z,X\}
\right]\\
&=0.
\end{align*}
The first $p$ columns and last $K$ columns of this zero matrix are respectively $\Cov(s_Y,s_\beta)$ and $\Cov(s_Y,s_T)$.
\end{proof}

\subsection*{Proof of Theorem~\ref{thm:rootfixed}}
\begin{proof}
Let $\Psi_n(\vartheta)$ denote the gradient of the criterion after
profiling out $\alpha$. The estimator satisfies
$\Psi_n(\widehat\vartheta_\lambda)=o_p(n^{-1/2})$, allowing for a
negligible optimization error. The expansion in
Assumption~\ref{ass:quad} gives
\[
0=\Psi_n(\vartheta_0)
  +\{\mathcal A_\lambda+o_p(1)\}
    (\widehat\vartheta_\lambda-\vartheta_0)
  +o_p(n^{-1/2}).
\]
Under equality, the penalty gradient is zero at $\vartheta_0$.
The scores in \eqref{eq:coxscore} and \eqref{eq:auxprofilescore} are
positive likelihood and residual scores, whereas $\Psi_n$ is the
gradient of a loss. Consequently,
\[
\sqrt n\,\Psi_n(\vartheta_0)
=-\frac{1}{\sqrt n}\sum_{i=1}^n
\begin{pmatrix}
s_\beta(O_i)\\s_T(O_i)\\s_Y(O_i)
\end{pmatrix}
+o_p(1).
\]
Assumption~\ref{ass:pd} and the nonnegative penalty Hessian imply that
$\mathcal A_\lambda$ is positive definite. Solving the preceding
expansion therefore yields
\[
\sqrt n(\widehat\vartheta_\lambda-\vartheta_0)
=\mathcal A_\lambda^{-1}
\frac{1}{\sqrt n}\sum_{i=1}^n
\begin{pmatrix}
s_\beta(O_i)\\s_T(O_i)\\s_Y(O_i)
\end{pmatrix}
+o_p(1),
\]
which is \eqref{eq:rootfull}. The joint central limit theorem, the score
covariance in \eqref{eq:B}, and Slutsky's theorem give the stated limit.
\end{proof}

\subsection*{Proof of Lemma~\ref{lem:eliminate}}
\begin{proof}
For a generic score contribution, write
$\delta=\mathcal A_\lambda^{-1}(s_\beta^\top,s_T^\top,s_Y^\top)^\top$
and partition it as $(\delta_\beta^\top,\delta_T^\top,\delta_Y^\top)^\top$.
It satisfies
\begin{align}
A\delta_\beta+C\delta_T&=s_\beta,
\label{eq:lin1}\\
C^\top\delta_\beta+(D+\lambda G)\delta_T-\lambda G\delta_Y&=s_T,
\label{eq:lin2}\\
-\lambda G\delta_T+(E+\lambda G)\delta_Y&=s_Y.
\label{eq:lin3}
\end{align}
From \eqref{eq:lin3},
\[
\delta_Y=(E+\lambda G)^{-1}(\lambda G\delta_T+s_Y).
\]
Substituting this expression into \eqref{eq:lin2} yields
\begin{align*}
C^\top\delta_\beta
&+\left[D+\lambda G-\lambda^2G(E+\lambda G)^{-1}G\right]\delta_T\\
&=s_T+\lambda G(E+\lambda G)^{-1}s_Y.
\end{align*}
The matrix multiplying $\delta_T$ is exactly $H_\lambda$, and the right-hand random term is $s_T+W_\lambda s_Y$. Together with \eqref{eq:lin1}, this proves \eqref{eq:Atilde}--\eqref{eq:stilde}.

For the covariance, Assumption~\ref{ass:crossorth} gives $\Cov(s_T,s_Y)=\Cov(s_\beta,s_Y)=0$. Therefore
\begin{align*}
\Var(s_T+W_\lambda s_Y)
&=\Var(s_T)+W_\lambda\Var(s_Y)W_\lambda^\top\\
&=D+W_\lambda E W_\lambda^\top
=J_\lambda,
\end{align*}
and
\[
\Cov(s_\beta,s_T+W_\lambda s_Y)=C.
\]
This proves \eqref{eq:Btilde}.
\end{proof}

\subsection*{Proof of Lemma~\ref{lem:psd}}
\begin{proof}
The case $\lambda=0$ is immediate, so suppose $\lambda>0$. Define
\[
L=E^{-1/2}GE^{-1/2}.
\]
Because $E$ and $G$ are positive definite, $L$ is positive definite. Also,
\[
E+\lambda G=E^{1/2}(I+\lambda L)E^{1/2},
\]
so
\[
(E+\lambda G)^{-1}
=E^{-1/2}(I+\lambda L)^{-1}E^{-1/2}.
\]
Define
\[
M_\lambda=\lambda L(I+\lambda L)^{-1}.
\]
Since $M_\lambda$ is a matrix function of the positive-definite symmetric matrix $L$, it is symmetric and has eigenvalues
\[
m_j=\frac{\lambda\ell_j}{1+\lambda\ell_j},
\]
where $\ell_j>0$ are the eigenvalues of $L$. Hence $0<m_j<1$.

We next rewrite $K_\lambda$. Since $G=E^{1/2}LE^{1/2}$,
\begin{align*}
K_\lambda
&=\lambda E^{1/2}LE^{1/2}
-\lambda^2E^{1/2}L(I+\lambda L)^{-1}LE^{1/2}\\
&=E^{1/2}
\left[\lambda L-\lambda^2L(I+\lambda L)^{-1}L\right]
E^{1/2}.
\end{align*}
Because $L$ commutes with $(I+\lambda L)^{-1}$,
\begin{align*}
\lambda L-\lambda^2L(I+\lambda L)^{-1}L
&=\lambda L\left[I-\lambda(I+\lambda L)^{-1}L\right]\\
&=\lambda L(I+\lambda L)^{-1}=M_\lambda.
\end{align*}
Thus
\begin{equation}
K_\lambda=E^{1/2}M_\lambda E^{1/2}\succeq0.
\label{eq:Kfactor}
\end{equation}

We now treat $W_\lambda E W_\lambda^\top$. Using the same factorization,
\begin{align*}
W_\lambda
&=\lambda E^{1/2}LE^{1/2}
E^{-1/2}(I+\lambda L)^{-1}E^{-1/2}\\
&=E^{1/2}M_\lambda E^{-1/2}.
\end{align*}
Therefore
\begin{align*}
W_\lambda E W_\lambda^\top
&=E^{1/2}M_\lambda E^{-1/2}E
E^{-1/2}M_\lambda E^{1/2}\\
&=E^{1/2}M_\lambda^2E^{1/2}.
\end{align*}
Combining this identity with \eqref{eq:Kfactor},
\[
\Delta_\lambda
=K_\lambda-W_\lambda EW_\lambda^\top
=E^{1/2}(M_\lambda-M_\lambda^2)E^{1/2}.
\]
The eigenvalues of $M_\lambda-M_\lambda^2$ are $m_j(1-m_j)\ge0$. Hence $\Delta_\lambda\succeq0$.
\end{proof}

\subsection*{Proof of Theorem~\ref{thm:efficiency}}
\begin{proof}
By Lemma~\ref{lem:eliminate}, the asymptotic covariance for $(\widehat\beta_\lambda,\widehat\theta_{T,\lambda})$ can be computed from
\[
\widetilde{\mathcal A}_\lambda^{-1}
\widetilde{\mathcal B}_\lambda
\widetilde{\mathcal A}_\lambda^{-1}.
\]
Since
\[
\widetilde{\mathcal B}_\lambda
=
\begin{pmatrix}A&C\\C^\top&H_\lambda-\Delta_\lambda\end{pmatrix}
=
\widetilde{\mathcal A}_\lambda
-
\begin{pmatrix}0&0\\0&\Delta_\lambda\end{pmatrix},
\]
we obtain the exact identity
\begin{align}
\widetilde{\mathcal A}_\lambda^{-1}
\widetilde{\mathcal B}_\lambda
\widetilde{\mathcal A}_\lambda^{-1}
&=\widetilde{\mathcal A}_\lambda^{-1}
-\widetilde{\mathcal A}_\lambda^{-1}
\begin{pmatrix}0&0\\0&\Delta_\lambda\end{pmatrix}
\widetilde{\mathcal A}_\lambda^{-1}.
\label{eq:sanddecomp}
\end{align}

We next compute the relevant blocks of $\widetilde{\mathcal A}_\lambda^{-1}$. The Schur complement of $H_\lambda$ in $\widetilde{\mathcal A}_\lambda$ is
\[
P_\lambda=A-CH_\lambda^{-1}C^\top.
\]
Because $\widetilde{\mathcal A}_\lambda$ is positive definite, $P_\lambda$ is positive definite. The block inverse formula gives
\begin{align}
(\widetilde{\mathcal A}_\lambda^{-1})_{\beta\beta}
&=P_\lambda^{-1},
\label{eq:blockbb}\\
(\widetilde{\mathcal A}_\lambda^{-1})_{\beta T}
&=-P_\lambda^{-1}CH_\lambda^{-1}.
\label{eq:blockbt}
\end{align}
Taking the $(\beta,\beta)$ block of \eqref{eq:sanddecomp} and using \eqref{eq:blockbb}--\eqref{eq:blockbt} yields \eqref{eq:Vclosed}.

By Lemma~\ref{lem:psd}, $\Delta_\lambda\succeq0$. Hence
\[
P_\lambda^{-1}CH_\lambda^{-1}\Delta_\lambda
H_\lambda^{-1}C^\top P_\lambda^{-1}\succeq0,
\]
so \eqref{eq:Vclosed} implies
\begin{equation}
V_\lambda\preceq P_\lambda^{-1}.
\label{eq:firstbound}
\end{equation}

Again by Lemma~\ref{lem:psd}, $K_\lambda\succeq0$, and therefore
\[
H_\lambda=D+K_\lambda\succeq D.
\]
For positive-definite matrices, inversion reverses the Loewner order; hence
\[
H_\lambda^{-1}\preceq D^{-1}.
\]
Premultiplying by $C$ and postmultiplying by $C^\top$ preserves positive-semidefinite order, giving
\[
CH_\lambda^{-1}C^\top\preceq CD^{-1}C^\top.
\]
Subtracting these matrices from $A$ reverses the order and gives
\[
P_\lambda=A-CH_\lambda^{-1}C^\top
\succeq
A-CD^{-1}C^\top=P_0.
\]
Both matrices are positive definite, so inversion again reverses order:
\begin{equation}
P_\lambda^{-1}\preceq P_0^{-1}=V_0.
\label{eq:secondbound}
\end{equation}
Combining \eqref{eq:firstbound} and \eqref{eq:secondbound}, and noting
that $V_\lambda$ is a covariance matrix and hence positive semidefinite,
proves \eqref{eq:Vorder}.
\end{proof}

\subsection*{Proof of Corollary~\ref{cor:strict}}
\begin{proof}
From \eqref{eq:Vclosed},
\begin{align*}
a^\top(P_\lambda^{-1}-V_\lambda)a
&=a^\top P_\lambda^{-1}CH_\lambda^{-1}
\Delta_\lambda H_\lambda^{-1}C^\top P_\lambda^{-1}a\\
&=\norm{\Delta_\lambda^{1/2}H_\lambda^{-1}C^\top P_\lambda^{-1}a}^2.
\end{align*}
Thus condition \eqref{eq:strict2} makes the first inequality $V_\lambda\preceq P_\lambda^{-1}$ strict in direction $a$. Condition \eqref{eq:strict1} makes the second inequality $P_\lambda^{-1}\preceq V_0$ strict in direction $a$. Either condition therefore yields $a^\top V_\lambda a<a^\top V_0a$.

If $C=0$, then $P_\lambda=A=P_0$ and the correction term in \eqref{eq:Vclosed} is zero, so $V_\lambda=A^{-1}=V_0$.
\end{proof}

\subsection*{Proof of Theorem~\ref{thm:local}}
\begin{proof}
Write
$\vartheta_{0,n}=(\beta_0^\top,\theta_{T0,n}^\top,\theta_{Y0,n}^\top)^\top$
for the moving unpenalized truth. At this parameter, each limiting
outcome-specific score has mean zero. The only nonzero deterministic first-order term arises from the discrepancy penalty. Since
\[
\theta_{T0,n}-\theta_{Y0,n}=-n^{-1/2}h,
\]
the population penalty gradient with respect to $(\beta,\theta_T,\theta_Y)$ is
\[
\begin{pmatrix}
0\\
\lambda G(\theta_{T0,n}-\theta_{Y0,n})\\
\lambda G(\theta_{Y0,n}-\theta_{T0,n})
\end{pmatrix}
=
\frac1{\sqrt n}
\begin{pmatrix}
0\\-\lambda Gh\\\lambda Gh
\end{pmatrix}.
\]
The empirical penalty uses $G_n$ in place of $G$. Since $G_n\to_p G$,
this replacement changes the root-$n$ penalty term by $o_p(1)$.
Let $s_{\beta,n}(O_i)$, $s_{T,n}(O_i)$, and $s_{Y,n}(O_i)$ denote the
score contributions at the moving truth. The root-$n$ estimating equation
is therefore
\[
0=
-\frac1{\sqrt n}\sum_{i=1}^n
\begin{pmatrix}s_{\beta,n}(O_i)\\s_{T,n}(O_i)\\s_{Y,n}(O_i)\end{pmatrix}
+
\begin{pmatrix}0\\-\lambda Gh\\\lambda Gh\end{pmatrix}
+
\mathcal A_\lambda
\sqrt n(\widehat\vartheta_\lambda-\vartheta_{0,n})
+o_p(1).
\]
The assumed common-limit convergence of the derivative and score
covariance matrices, together with the uniform local expansion, gives
the limiting covariance $\mathcal B$. The additional first-order term
is the displayed deterministic penalty gradient.

Let $b=(b_\beta^\top,b_T^\top,b_Y^\top)^\top$ denote the deterministic mean shift of $\sqrt n(\widehat\vartheta_\lambda-\vartheta_{0,n})$. Taking expectations of the limiting linear system gives
\begin{equation}
\mathcal A_\lambda b
=-
\begin{pmatrix}0\\-\lambda Gh\\\lambda Gh\end{pmatrix}
=
\begin{pmatrix}0\\\lambda Gh\\-\lambda Gh\end{pmatrix}.
\label{eq:meansystem}
\end{equation}
Writing the three block equations explicitly,
\begin{align}
Ab_\beta+Cb_T&=0,
\label{eq:b1}\\
C^\top b_\beta+(D+\lambda G)b_T-\lambda Gb_Y&=\lambda Gh,
\label{eq:b2}\\
-\lambda Gb_T+(E+\lambda G)b_Y&=-\lambda Gh.
\label{eq:b3}
\end{align}
From \eqref{eq:b3},
\[
b_Y=(E+\lambda G)^{-1}(\lambda Gb_T-\lambda Gh).
\]
Substituting this into \eqref{eq:b2} gives
\begin{align*}
C^\top b_\beta
&+\{D+\lambda G-\lambda^2G(E+\lambda G)^{-1}G\}b_T\\
&=\lambda Gh-\lambda^2G(E+\lambda G)^{-1}Gh.
\end{align*}
The left-hand nuisance matrix is $H_\lambda$. The right-hand side is
\[
\{\lambda G-\lambda^2G(E+\lambda G)^{-1}G\}h
=K_\lambda h.
\]
Hence
\begin{equation}
C^\top b_\beta+H_\lambda b_T=K_\lambda h.
\label{eq:bprofile}
\end{equation}
Solving \eqref{eq:bprofile} for $b_T$ gives
\[
b_T=H_\lambda^{-1}(K_\lambda h-C^\top b_\beta).
\]
Substitute this into \eqref{eq:b1}:
\begin{align*}
0
&=Ab_\beta
+CH_\lambda^{-1}(K_\lambda h-C^\top b_\beta)\\
&=\{A-CH_\lambda^{-1}C^\top\}b_\beta
+CH_\lambda^{-1}K_\lambda h\\
&=P_\lambda b_\beta+CH_\lambda^{-1}K_\lambda h.
\end{align*}
Therefore
\[
b_\beta=-P_\lambda^{-1}CH_\lambda^{-1}K_\lambda h,
\]
which is \eqref{eq:Blambda}. The centered random part is the same linear transformation of the limiting Gaussian score as under $f_{Y0}=f_{T0}$, so its covariance is $V_\lambda$. Finally, $K_0=0$, implying $B_0(h)=0$.
\end{proof}

\subsection*{Proof of Proposition~\ref{prop:vanish}}
\begin{proof}
Let $\Psi_0(\vartheta)$ denote the gradient of the unpenalized population criterion and let $p(\vartheta)$ denote the gradient of the unit-weight discrepancy penalty. The penalized population first-order condition is
\begin{equation}
\Psi_0(\vartheta_\lambda)+\lambda p(\vartheta_\lambda)=0.
\label{eq:poptargeteq}
\end{equation}
At the unpenalized truth, $\Psi_0(\vartheta_0)=0$. Let $H_0=\dot\Psi_0(\vartheta_0)$, which is nonsingular by assumption. A Taylor expansion gives
\[
\Psi_0(\vartheta_\lambda)
=H_0(\vartheta_\lambda-\vartheta_0)
+o(\norm{\vartheta_\lambda-\vartheta_0}).
\]
Continuity of the penalty gradient gives
\[
p(\vartheta_\lambda)=p(\vartheta_0)+o(1).
\]
Substitution into \eqref{eq:poptargeteq} yields
\[
H_0(\vartheta_\lambda-\vartheta_0)
=-\lambda p(\vartheta_0)
+o(\norm{\vartheta_\lambda-\vartheta_0})+o(\lambda).
\]
Because $H_0^{-1}$ is bounded, the implicit-function theorem, or the displayed expansion itself after absorbing the smaller-order term, gives
\[
\vartheta_\lambda-\vartheta_0
=-\lambda H_0^{-1}p(\vartheta_0)+o(\lambda),
\]
proving \eqref{eq:poptargetrate}. If $\lambda_n\to0$ and the estimator
is consistent for $\vartheta_{\lambda_n}$, the triangle inequality gives
consistency for $\vartheta_0$. For the first-order distribution, when
$\sqrt n\lambda_n\to0$,
\[
\sqrt n(\vartheta_{\lambda_n}-\vartheta_0)
=O(\sqrt n\lambda_n)=o(1).
\]
Hence, if the stochastic expansion around $\vartheta_{\lambda_n}$ is uniform for $\lambda_n\to0$, centering at $\vartheta_0$ instead changes the root-$n$ statistic by $o_p(1)$.
\end{proof}

\subsection*{Proof of Theorem~\ref{thm:adaptive}}
\begin{proof}
Under equality, \eqref{eq:Drate} and \eqref{eq:taurate} give
\[
\widehat D_n/\tau_n=O_p(r_n/\tau_n)=o_p(1).
\]
Because $\omega$ is identically one on a neighborhood of zero,
$\Pp(\widehat\lambda_n=\lambda_+)\to1$, proving part (i).

Under a fixed difference, the reverse triangle inequality gives
\[
\left|\widehat D_n-\norm{f_{Y0}-f_{T0}}_n\right|
\le
\norm{\widehat f_Y^{\mathrm{sep}}-f_{Y0}}_n
+
\norm{\widehat f_T^{\mathrm{sep}}-f_{T0}}_n
=o_p(1).
\]
Square-integrability of the fixed difference and the law of large
numbers give $\norm{f_{Y0}-f_{T0}}_n\to_p D_0>0$. Hence
$\widehat D_n/\tau_n\to_p\infty$. Since $\omega(u)=0$ for $u\ge1$,
$\Pp(\widehat\lambda_n=0)\to1$, proving part (ii).

These probability-one-in-the-limit selection statements also imply
part (iii). Use the same measurable rule to choose a minimizer at each
penalty level. On the event $\widehat\lambda_n=\lambda_+$ the adaptive
and fixed-penalty estimators are identical, and on the event
$\widehat\lambda_n=0$ the adaptive and separate estimators are identical.
The probability of the complementary event tends to zero in the
respective cases, which gives both stated $o_p(1)$ differences. In fact,
the flat regions of $\omega$ make the additional equicontinuity condition
in part (iii) unnecessary for these two pointwise conclusions.
\end{proof}

\subsection*{Proof of Theorem~\ref{thm:growing}}
\begin{proof}
Under Assumptions~\ref{ass:nuisancerate}--\ref{ass:uniformblocks}, the profiled target estimator has the same root-$n$ linear representation as the finite-sieve block system up to $o_p(1)$. Part (i) therefore follows from the limiting central limit theorem and Slutsky's theorem.

For part (ii), let $a\in\R^p$. The finite-sieve ordering gives
\[
a^\top V_{\lambda,n}a\le a^\top V_{0,n}a
\]
for every $n$. Taking limits yields
\[
a^\top V_\lambda a\le a^\top V_0a.
\]
Because this holds for every $a$, $V_\lambda\preceq V_0$.

For part (iii), the local deterministic penalty term enters the sieve
estimating equations as in Theorem~\ref{thm:local}. Its target component
is $B_{\lambda,n}(h_n)$, whose convergence to $B_\lambda$ is assumed in
the theorem. The representation error and profile remainder are
negligible by the stated assumptions. Combining the mean-shift limit
with the centered limit from part (i) gives $N(B_\lambda,V_\lambda)$.
\end{proof}

\section{Matrix details for the variance comparison}
\label{app:algebra}

\subsection{Positive definiteness}

\begin{lemma}
Under Assumption~\ref{ass:pd}, $\mathcal A_\lambda$ in \eqref{eq:Alambda} and $\widetilde{\mathcal A}_\lambda$ in \eqref{eq:Atilde} are positive definite for every $\lambda\ge0$.
\end{lemma}

\begin{proof}
For any $(u^\top,v^\top,w^\top)^\top$,
\begin{align*}
\begin{pmatrix}u\\v\\w\end{pmatrix}^\top
\mathcal A_\lambda
\begin{pmatrix}u\\v\\w\end{pmatrix}
&=
\begin{pmatrix}u\\v\end{pmatrix}^\top
\mathcal I_T
\begin{pmatrix}u\\v\end{pmatrix}
+w^\top Ew
+\lambda(v-w)^\top G(v-w).
\end{align*}
The first two terms are strictly positive unless $(u,v)=(0,0)$ and $w=0$, respectively, and the last term is nonnegative. Hence the sum is positive for every nonzero vector. Therefore $\mathcal A_\lambda\succ0$. A Schur complement of a positive-definite block matrix is positive definite, so eliminating the $\theta_Y$ block implies $\widetilde{\mathcal A}_\lambda\succ0$.
\end{proof}

\subsection{Large-penalty limit}

\begin{proposition}
Assume $E$ and $G$ are positive definite. Then
\[
K_\lambda\longrightarrow E
\qquad\text{as }\lambda\to\infty.
\]
Consequently,
\[
H_\lambda\to D+E.
\]
\end{proposition}

\begin{proof}
From \eqref{eq:Kfactor},
\[
K_\lambda=E^{1/2}M_\lambda E^{1/2},
\qquad
M_\lambda=\lambda L(I+\lambda L)^{-1}.
\]
For each eigenvalue $\ell_j>0$ of $L$,
\[
\frac{\lambda\ell_j}{1+\lambda\ell_j}\to1.
\]
Hence $M_\lambda\to I$, so $K_\lambda\to E$. The second conclusion follows from $H_\lambda=D+K_\lambda$.
\end{proof}

Under strong penalization, the effective nuisance curvature approaches
$D+E$, as in a fit that imposes equality of the two nuisance coefficients.
Under the information identities and score orthogonality of
Theorem~\ref{thm:efficiency}, $W_\lambda\to I$ and
$\Delta_\lambda\to0$, so
\[
V_\lambda\longrightarrow
\{A-C(D+E)^{-1}C^\top\}^{-1}.
\]
Thus the correction term in \eqref{eq:Vclosed} vanishes in this limit
under those assumptions. A general pooled fit without the information
identities or score orthogonality still requires its full sandwich
covariance.

\section{Details for the growing-sieve argument}
\label{app:sieve}

This appendix explains the high-level conditions used in
Theorem~\ref{thm:growing}. It does not verify them for a particular spline
basis. Such a verification requires explicit assumptions on smoothness,
basis growth, and the covariate and risk-set distributions.

\subsection{Uniform convergence of the criterion}

Let $\Theta_n$ be an identifiable sieve parameter space for
$\vartheta=(\beta^\top,\theta_T^\top,\theta_Y^\top)^\top$ and define the
profiled criteria
\[
Q_n^{\mathrm{pr}}(\vartheta;\lambda)
=\inf_{\alpha\in\R^p}
Q_n\bigl(\alpha,\beta,
 b_{K_n}(\cdot)^\top\theta_Y,
 b_{K_n}(\cdot)^\top\theta_T;\lambda\bigr),
\]
with $Q^{\mathrm{pr}}$ defined analogously from the population criterion
\eqref{eq:popQ}. A consistency argument requires a uniform law of large
numbers on a suitably restricted parameter space, of the form
\[
\sup_{\vartheta\in\Theta_n}
\left|Q_n^{\mathrm{pr}}(\vartheta;\lambda)
      -Q^{\mathrm{pr}}(\vartheta;\lambda)\right|=o_p(1),
\]
together with appropriate separation of the population minimizer.
This display is a condition to be established, not a consequence of
sieve approximation alone.

For the Cox component, write
\[
\eta_\vartheta(Z,X)
=Z^\top\beta+b_{K_n}(X)^\top\theta_T,
\qquad
W_n=(Z^\top,b_{K_n}(X)^\top)^\top.
\]
The risk-set moments to be controlled are
\[
S_n^{(r)}(t;\vartheta)
=\Pn\{R(t)e^{\eta_\vartheta(Z,X)}W_n^{\otimes r}\},
\qquad r=0,1,2.
\]
For the penalty, the difference between the empirical and population
quadratic forms is
\[
\frac{\lambda}{2}
(\theta_Y-\theta_T)^\top
\{G_n-G^{(n)}\}(\theta_Y-\theta_T).
\]
Its control requires bounds on both the Gram-matrix difference and the
coefficient vectors over the chosen parameter space.

\subsection{Consistency}

Uniform convergence and identification yield consistency for the
relevant population minimizer. Under equality, the unrestricted
population target is the unpenalized truth by
Proposition~\ref{prop:exacttarget}; sieve approximation must also be
accounted for. Under local alternatives, consistency is to the common
limiting model. The condition
$\theta_{Y0,n}-\theta_{T0,n}=n^{-1/2}h$ does not by itself imply that
each nuisance coefficient is within $O(n^{-1/2})$ of its common limit.
The first-order target shift instead follows from the projected expansion
and the assumptions of Theorem~\ref{thm:growing}. Under a fixed
difference and fixed positive penalty, the population target is generally
the penalized minimizer rather than the unpenalized truth.

\subsection{Rate of the nuisance estimators}

A rate argument must control both sieve approximation and stochastic
error in the joint criterion. Assumption~\ref{ass:nuisancerate} requires
that these errors give an $o_p(n^{-1/2})$ remainder in the profiled target
equation. When the estimation remainder is quadratic, the condition
$r_n=o(n^{-1/4})$ controls that part of the error. A separate bound is
needed for the approximation bias in the target equation. Explicit
restrictions on $K_n$ depend on the basis, smoothness, and the
empirical-process bounds used for the Cox risk sets.

\subsection{Quadratic expansion}

For a reference sieve parameter $\vartheta_{0,n}$ and an admissible
increment $\delta$, a local Taylor expansion of the profiled criterion
has the form
\begin{align*}
&Q_n^{\mathrm{pr}}(\vartheta_{0,n}+\delta;\lambda)
 -Q_n^{\mathrm{pr}}(\vartheta_{0,n};\lambda)\\
&\quad=
\nabla Q_n^{\mathrm{pr}}(\vartheta_{0,n};\lambda)^\top\delta
+\frac12\delta^\top
 \nabla^2 Q_n^{\mathrm{pr}}(\vartheta_{0,n};\lambda)\delta
+R_n(\delta).
\end{align*}
For a correctly specified fixed sieve under equality, the linear term
uses the negative of the score vector defined in
Section~\ref{sec:firstorder}, as shown in the proof of
Theorem~\ref{thm:rootfixed}. Under local differences, the penalty gradient
must also be retained.

The condition needed for Theorem~\ref{thm:growing} concerns the remainder
after profiling the target equation. It is not a blanket assertion that
$\sup_{\|\delta\|\le c r_n}|R_n(\delta)|=o_p(n^{-1})$ over the entire
nuisance neighborhood; such a bound can require stronger rates.
The penalty is exactly quadratic in the sieve coefficients. Its
empirical Gram matrix is $G_n$, and its population counterpart at
dimension $K_n$ is $G^{(n)}$.

\subsection{Profiling the nuisance coordinates}

Partition the estimating equations into the target coordinate $\beta$
and the nuisance coordinates. Solving the nuisance linearized equations
and substituting into the target equation gives the profile score and
its remainder. The relevant derivative and covariance blocks are those
in Sections~\ref{sec:firstorder}--\ref{sec:local} at dimension $K_n$.
Assumptions~\ref{ass:nuisancerate}--\ref{ass:uniformblocks} require that
the resulting approximation error be negligible in the root-$n$ target
expansion.

\subsection{Passage to the limit}

Once the projected score satisfies the required central limit theorem
and $V_{\lambda,n}$ and $V_{0,n}$ converge, the finite-sieve inequality
passes to the limit. Indeed, for each $a\in\R^p$,
\[
a^\top V_{\lambda,n}a\le a^\top V_{0,n}a
\quad\Longrightarrow\quad
a^\top V_\lambda a\le a^\top V_0a.
\]
For local alternatives, the additional deterministic limit is
$\lim_n B_{\lambda,n}(h_n)$, as assumed in
Theorem~\ref{thm:growing}.

\section{Allowing an outcome-specific scale parameter}
\label{app:rho}

A proportional extension replaces the discrepancy penalty by
\[
\frac\lambda2\norm{f_Y-\rho f_T}_n^2.
\]
Here $\rho$ is a proportionality parameter. Rescaling $Y$ by a nonzero
constant $c$ rescales this squared discrepancy by $c^2$ when $f_Y$ and
$\rho$ are rescaled accordingly; preserving the numerical criterion
therefore also requires replacing $\lambda$ by $\lambda/c^2$ and
rescaling the auxiliary residual variance. Proportionality alone is
not invariance at a fixed penalty level.

If $f_{Y0}=\rho_0f_{T0}$, the first derivative of the penalty vanishes
at the true functions and $\rho=\rho_0$. If $\rho_0$ is treated as fixed, the nuisance penalty Hessian for $(\theta_T,\theta_Y)$ becomes
\[
\lambda
\begin{pmatrix}
\rho_0^2G&-\rho_0G\\
-\rho_0G&G
\end{pmatrix},
\]
and the elimination argument can be repeated with modified blocks. If $\rho$ is estimated, an additional row and column must be included in the derivative and covariance matrices. Identification of $\rho$ additionally requires a nonzero survival
covariate function; if $f_{T0}=0$, the equality restriction alone does
not identify $\rho$. Regularity of this expanded system also requires
separate analysis. For that reason, the present paper uses a prespecified outcome standardization in the main theory and treats unknown proportional scaling as an extension.

\bibliographystyle{plainnat}
\bibliography{references}

\end{document}